\documentclass[runningheads]{llncs}

\newif\ifTechnicalReport
\TechnicalReporttrue

\usepackage[UKenglish]{babel}
\usepackage[utf8]{inputenc}
\usepackage[T1]{fontenc}

\usepackage{graphicx}

\usepackage{amsmath, amssymb}
\usepackage{bm}
\usepackage{mathtools}
\mathtoolsset{centercolon} \usepackage{thm-restate}

\usepackage[ruled, linesnumbered, noend]{algorithm2e} 

\usepackage{xcolor}
\usepackage{tikz}
\usetikzlibrary{arrows.meta}

\usepackage{enumitem}

\usepackage{todonotes}

\usepackage[hypertexnames=false]{hyperref} \newcommand{\ourtitle}{Stratified Negation in RDF Rules:\texorpdfstring{\\}{ }A Correct Approach}

\ifTechnicalReport
\let\ouroriginaltitle\ourtitle
\renewcommand{\ourtitle}{\ouroriginaltitle~(Extended Version)\texorpdfstring{\thanks{This is the technical report accompanying our ISWC26 paper \cite{KIAK:chains}.}}{}}
\fi

\newcommand{\arxivLink}{https://arxiv.org/XXX}
\newcommand{\gitlabLink}{https://gitlab.com/nkuechen/chaining-reliances/-/tree/iswc26}

\hypersetup{hidelinks,pdftitle={\ourtitle},pdfauthor={N. Küchenmeister, A. Ivliev, D. Arndt, M. Krötzsch}} 

\newcommand{\opfont}[1]{\text{{\sf #1}}} \newcommand{\tuple}[1]{\langle{#1}\rangle}
\renewcommand{\vec}[1]{\boldsymbol{#1}}

\def\set[#1]{\ensuremath{\{1,\ldots,#1\}}} 
 \newcommand{\id}{\text{id}} 
\newcommand{\injTo}{\to} \newcommand{\ruleTo}{\to} 
\newcommand{\Pred}{\ensuremath{\bm{P}}} \newcommand{\ar}{\opfont{ar}} 

\newcommand{\V}{\ensuremath{\bm{V}}} \newcommand{\C}{\ensuremath{\bm{C}}} \newcommand{\N}{\ensuremath{\bm{N}}} 
\def\var(#1){\ensuremath{\opfont{var}(#1)}} \def\varA(#1){\ensuremath{\opfont{var}_\forall(#1)}} \def\varE(#1){\ensuremath{\opfont{var}_\exists(#1)}} 
\def\bodyP(#1){\ensuremath{\opfont{body}^+(#1)}} \def\bodyN(#1){\ensuremath{\opfont{body}^-(#1)}} \def\head(#1){\ensuremath{\opfont{head}(#1)}} 
\def\pieces(#1){\ensuremath{\opfont{pieces}(#1)}} 
\newcommand{\I}{\ensuremath{\mathcal{I}}} \newcommand{\R}{\ensuremath{\mathcal{R}}} \newcommand{\KB}{\ensuremath{\langle\I,\R\rangle}} 
\newcommand{\RpD}{\ensuremath{\R_D}} 
\def\matches(#1){\ensuremath{\opfont{matches}_\R(#1)}} \def\unsat(#1){\ensuremath{\opfont{unsat}_\R(#1)}} 
\def\symbols(#1){\ensuremath{\opfont{symbols}(#1)}} 
\def\cDB(#1,#2){\ensuremath{\opfont{db}_{#1}(#2)}} \def\cRule(#1,#2){\ensuremath{\opfont{rule}_{#1}(#2)}} \def\cHom(#1,#2){\ensuremath{\opfont{hom}_{#1}(#2)}} \def\cMatch(#1,#2){\ensuremath{\opfont{match}_{#1}(#2)}} \def\cSteps(#1){\ensuremath{\opfont{steps}_{#1}}} 
\def\chase{\ensuremath{\opfont{chase}}} 
\def\ins{\ensuremath{\hat{\rho}}}

\def\stale(#1){\ensuremath{\opfont{stale}(#1)}}
\def\decoupled(#1){\ensuremath{\opfont{decoupled}(#1)}} 
\newcommand{\posr}{\ensuremath{\prec^+}} \newcommand{\restr}{\ensuremath{\prec^\square}} \newcommand{\negr}{\ensuremath{\prec^-}} 
\newcommand{\pprec}{\ensuremath{\prec\!\!\!\prec}}
\newcommand{\iposr}{\ensuremath{\pprec^+}} 
\newcommand{\posrT}{\ensuremath{\prec_t^+}}
\newcommand{\posrC}{\ensuremath{\prec_c^+}}

\newcommand{\restrC}{\ensuremath{\prec_c^\square}} \newcommand{\negrC}{\ensuremath{\prec_c^-}} 
\newcommand{\restrCD}{\ensuremath{\prec^{\square D}_c}} \newcommand{\negrCD}{\ensuremath{\prec^{- D}_c}} 

\newcommand{\rtHull}[1]{\ensuremath{#1^{\bm{\star}}}} 
\def\variants(#1){\ensuremath{\opfont{variants}(#1)}} \def\instances(#1){\ensuremath{\opfont{instances}(#1)}} 
\def\subformulae(#1){\ensuremath{\opfont{parts}(#1)}} 
\newcommand{\orig}{\opfont{orig}}  
\def\chains(#1){\ensuremath{\opfont{chains}(#1)}} 
\newcommand{\Alphabet}{\ensuremath{\Sigma}} \newcommand{\labAname}{\opfont{C}} \newcommand{\labBname}{\opfont{S}} \newcommand{\labCname}{\opfont{N}} \newcommand{\labA}{\ell^{\labAname}} \newcommand{\labB}{\ell^{\labBname}} \newcommand{\labC}{\ell^{\labCname}} 
\def\L(#1){\ensuremath{\mathcal{L}_{#1}}} 
\def\pre{\ensuremath{\opfont{pre}}} \def\rep(#1){\ensuremath{\opfont{rep}(#1)}} 
\def\fresh{\ensuremath{\opfont{fresh}}} \def\truncation(#1){\ensuremath{\opfont{truncation}(#1)}}  \newcommand{\myparagraph}[1]{\vspace{1ex plus 0.5ex minus 0.5ex}\noindent\textbf{#1}~~}
\newcommand{\alglineref}[1]{L\ref{#1}} \newcommand{\mintlabel}[1]{\phantomsection\label{#1}} 

\makeatletter
\let\c@lemma\c@theorem
\let\c@corollary\c@theorem
\makeatother

\newcommand{\cont}[1]{cont.}
\newcommand{\contof}[1]{cont.\ of Ex.~\ref{#1}} 

\usepackage{anyfontsize}
\usepackage{minted}
\setminted{escapeinside=||,numbersep=6pt}

\definecolor{RuleBgColor}{RGB}{240,240,240}
\newenvironment{nthreelisting}{\VerbatimEnvironment\begin{minted}[fontsize=\fontsize{7.5}{9}\selectfont, bgcolor=RuleBgColor, linenos, firstnumber=last, numbersep=3pt, xleftmargin=9pt]{sparql}}
    {\end{minted}}
\newmintinline[nthreeinline]{sparql}{}

\usepackage{orcidlink}
\renewcommand{\orcidID}[1]{\,\orcidlink{#1}}

\begin{document}

\title{\ourtitle}

\author{
    Nils Küchenmeister\inst{1}\orcidID{0009-0004-0376-0328}\and Alex Ivliev\inst{1}\orcidID{0000-0002-1604-6308}\and Dörthe Arndt\inst{2}\orcidID{0000-0002-7401-8487}\and Markus~Krötzsch\inst{1}\orcidID{0000-0002-9172-2601}}

\authorrunning{N. Küchenmeister \and A. Ivliev \and D. Arndt \and M. Krötzsch}

\institute{
    Knowledge-Based Systems Group, TU Dresden, Dresden, Germany \and
    Computational Logic Group, TU Dresden / ScaDS.AI, Dresden/Leipzig, Germany
\email{firstname.lastname@tu-dresden.de}
}

\maketitle

\begin{abstract}
Combining RDF rule languages, such as N3 or SHACL Rules, with default negation is challenging.
Existing methods to stratify negation often fail for RDF rules, since individual triples do not
carry enough information to meaningfully restrict potential dependencies.
Blank nodes in rule heads further complicate the matter,
since the order of rule applications may determine whether new values are created,
which in turn can change the applicability of rules with negation.
To solve these open problems, we propose \emph{chain stratification} as a robust new condition that guarantees a
well-behaved semantics for RDF rules with negation,
and existential rules in general.
Our condition combines an elaborate analysis of potential multistep derivations with a
mechanism for using integrity constraints to discard impossible cases.
Applying rules in any order that respects chain stratification is guaranteed to
derive an RDF graph that is unique, lean, and justified under the usual negation-as-failure semantics.
To show the practicality, we also provide a prototype implementation.
\end{abstract}
 
\section{Introduction}\label{sec:intro}

The usefulness of rule languages for semantic web query answering, knowledge graph analysis, and ontological
reasoning has long been recognised, and various rule engines now support RDF and SPARQL \cite{N+15:RDFoxToolPaper,BSG:Vadalog18,Ivliev+:Nemo2024}.
Attempts to establish corresponding web standards have been less successful \cite{SWRL,RIF-overview},
but ongoing works on N3 \cite{n3-community-draft} and SHACL Rules \cite{shacl12-rules} may yet change this.
Even without standardised syntax, today's rule languages share a common semantic
basis --- the classical language Datalog \cite{Kroetzsch2025:Datalog} --- which is intuitive and easy to implement,
typically simply by applying rules until nothing new follows.

Some useful features introduce complications though. A first case is that of \emph{non-monotonic} operations, such as
\emph{negation as failure} and \emph{aggregation}, the results of which may become invalid in the light of new data. A standard solution to
this problem is to \emph{stratify} (i.e., ``order'') computation, so that non-monotonic operations are only applied
to data that has been fully computed \cite{Kroetzsch2025:Datalog}. This precludes feedback cycles from non-monotonic outputs to their
own inputs.
A second challenge is \emph{value invention} where rules may lead to the creation of new elements, such that an infinite
amount of new statements can be derived. In RDF, this is characteristic of rules with \emph{blank nodes} (aka \emph{bnodes}) in their conclusion;
in databases and logic, existential quantifiers are used instead \cite{BLMS11:decline}.
A first mitigation is to create new bnodes only if the rule is not already satisfied by another element in place of the bnode,
a method known as \emph{restricted} (or \emph{standard}) \emph{chase} \cite{DNR08:corechase}.
Non-termination may still occur, and analysing rules to preclude (potential) feedback cycles can help
detect such problems \cite{BLMS11:decline,CG+13:acyclicity}.

The working drafts for N3 \cite{n3-community-draft} and SHACL Rules \cite{shacl12-rules} combine 
non-monotonic features and value invention with a triple-based data model. Unfortunately, this leads to two
major problems: (1) existing stratification and termination analyses fail in many cases,
and (2) even if stratification succeeds, the rules may not have a well-defined semantics.
Solving these will be our main contribution, but it is useful to first understand each problem in more detail.

\subsubsection{Problem 1: Stratification fails with RDF rules.}
Traditional stratification is based on analysing dependencies between different \emph{predicates} used in the rules.
RDF does not have predicates, or, equivalently, has only a single predicate \emph{triple} of arity three,\footnote{Such ternary predicates are also how relational rule engines import RDF \cite{Ivliev+:Nemo2024}.
Treating property names as binary predicates to represent RDF rules in Datalog
would prevent the use of variables in predicate positions, as required for the RDFS rules.
}
which means that any use of negation prevents stratification.
SHACL proposes a refined unification-based dependency analysis \cite[Section~3.4 ``Stratification'']{shacl12-rules}, which 
helps when triples contain constants, as these N3 rules:\footnote{We use prefixes \texttt{rdf:}, \texttt{rdfs:}, and \texttt{log:} for the standard namespaces \cite{rdf11-concepts,n3-community-draft}.}
\begin{nthreelisting}
{?c :participant ?p} => {?p rdf:type :Student}. |\mintlabel{rule_ex_intro_student}|
{?c :teacher ?p . [] log:notIncludes {?p rdf:type :Student}} => {?c :examiner ?p}.  |\mintlabel{rule_ex_intro_negation}|
\end{nthreelisting}
The rule in \alglineref{rule_ex_intro_student} states that every participant \nthreeinline{?p} of some course \nthreeinline{?c} is a student.
The rule in \alglineref{rule_ex_intro_negation} means that course teachers who are not students are also examiners for the course,
where \nthreeinline{[] log:notIncludes {|$\cdots$|}} is N3 syntax for negation.
SHACL's stratification method orders \alglineref{rule_ex_intro_student} before \alglineref{rule_ex_intro_negation}, so that
the negative precondition is only evaluated after inferring who is a student.
Unfortunately, the approach breaks when adding further rules. For example, the following rules \emph{rdfs5} and \emph{rdfs7}
are a standard way for handling subproperties in RDFS \cite{rdf11-semantics}:\begin{nthreelisting}
{?p rdfs:subPropertyOf ?q . ?q rdfs:subPropertyOf ?r} => {?p rdfs:subPropertyOf ?r}.  |\mintlabel{rule_ex_intro_rdfs5}|
{?p rdfs:subPropertyOf ?q . ?x ?p ?y} => {?x ?q ?y}.   |\mintlabel{rule_ex_intro_rdfs7}|
\end{nthreelisting}
The all-variable triples in rule \emph{rdfs7} (\alglineref{rule_ex_intro_rdfs7}) create dependencies from and to any other rule, making stratification
impossible.
Indeed, using all the rules \alglineref{rule_ex_intro_student}--\alglineref{rule_ex_intro_rdfs7}, there really are negative feedback cycles
on valid input data, e.g., with the triple\\ \nthreeinline{:examiner rdfs:subPropertyOf :participant.} In this case, a student might be discovered only
after applying rule \alglineref{rule_ex_intro_negation}. Similar issues would arise if \nthreeinline{:examiner} were a subproperty of
\nthreeinline{rdfs:subPropertyOf} or \nthreeinline{rdf:type}.
These cases are unintended and would likely indicate errors, which N3 can
express using \emph{constraints}:\begin{nthreelisting}
{:examiner rdfs:subPropertyOf :participant} => false.  |\mintlabel{rule_constraint_teacher_participant}|
{:examiner rdfs:subPropertyOf rdfs:subPropertyOf} => false. |\mintlabel{rule_constraint_teacher_spo}|
{:examiner rdfs:subPropertyOf rdf:type} => false. |\mintlabel{rule_constraint_teacher_type}|
\end{nthreelisting}
Negative feedback cycles are not possible without violating some of these constraints (which N3 tools would flag as an error),
so the use of negation is safe now. Unfortunately, current stratification methods cannot detect this.

\subsubsection{Problem 2: Stratified RDF rules lack semantics.}
Bnodes require the existence of suitable elements, which may or may not need to be created.
The next rule, e.g., states that humans must have some biological father who is male:
\begin{nthreelisting}
{?x rdf:type :Human} => {?x :father _:1 . _:1 rdf:type :Man}. |\mintlabel{rule_father}|
\end{nthreelisting}
Given (a) \nthreeinline{:jo rdf:type :Human}, applying \alglineref{rule_father} creates a new value
for \nthreeinline{_:1} (a fresh bnode in RDF, a \emph{labelled null} in databases),
but this should not happen if we already have (b) \nthreeinline{:jo :father :bob} and (c) \nthreeinline{:bob rdf:type :Man}.
However, it is not always so easy to tell if a fresh bnode is necessary, and the decision 
can critically affect the result, as can be shown with some further rules \cite[Ex.~6]{KR20:cores}:
\begin{nthreelisting}
{?x :father ?y} => {?y rdf:type :Man}. |\mintlabel{rule_father_male}|
{?x :father ?y} => {?y :eq ?y}. |\mintlabel{rule_father_selfequal}|
{?x :father ?y1. ?x :father ?y2. [] log:notIncludes {?y1 :eq ?y2}} => {?y1 :nef ?y2}. |\mintlabel{rule_different_fathers}|
\end{nthreelisting}
\alglineref{rule_father_male} defines a range for \nthreeinline{:father}.
\alglineref{rule_father_selfequal} derives a simple equality, which is negated in \alglineref{rule_different_fathers}
to check for non-equal fathers (\nthreeinline{:nef}).
Given triples (a) and (b), we could apply \alglineref{rule_father} to create a fresh bnode \nthreeinline{_:a}
(\nthreeinline{:bob} does not satisfy the conclusion at this point).
Rules \alglineref{rule_father_male} and \alglineref{rule_father_selfequal} would then produce (c), 
\nthreeinline{:bob :eq :bob} and \nthreeinline{_:a :eq _:a}, and \alglineref{rule_different_fathers}
would yield \nthreeinline{:bob :nef _:a} and \nthreeinline{_:a :nef :bob}.
In contrast, if we apply \alglineref{rule_father_male} first, then \alglineref{rule_father} will not apply,
and no \nthreeinline{:nef}-triple follows.

The issue is serious: classically stratified RDF rule sets do not have a well-defined semantics.
Indeed, \alglineref{rule_father}--\alglineref{rule_different_fathers} are stratified (by SHACL's method),
but stratification allows both of the above rule precedences.
Other non-monotonic features, such as inequality or aggregation, can also expose the problem.

We might hope to avoid this issue by interleaving rule applications with additional algorithms to find and delete redundant bnodes,
creating a \emph{lean} RDF graph \cite{Hogan:LeanRDF:Tweb17} (or \emph{core} \cite{DNR08:corechase}).
But even this (costly) workaround does not solve the problem: we can just apply \alglineref{rule_father_male} last.
As another solution, Krötzsch proposed a new type of dependency for rules with bnodes (existential variables) in their conclusion \cite{KR20:cores}.
His method can find a rule application order that leads to a lean graph and a well-defined result,
but it also fails with rules like rdfs7 (\alglineref{rule_ex_intro_rdfs7}).

\subsubsection{Our proposed solution.}
We define a precedence relation on rules that can guide the order of rule applications
so that negation-as-failure is justified and the inferred RDF graph is unique and lean (if the input data was).
We start from Krötzsch's \emph{reliances} \cite{KR20:cores}, recalled in
the preliminaries (Section~\ref{sec:prelim}), which define several types of binary relations between rules
to estimate potential interactions.

We generalise this approach based on two observations:
(1) the transitive composition of reliances significantly overestimates which (problematic) interactions can occur in real sequences of rule applications,
and
(2) constraints and plain Datalog rules can effectively rule out problematic cases that would otherwise be allowed in RDF.
The difficulty is to turn these into an implementable criterion, in particular (1), since there are infinitely many 
potential sequences of rules. In Section~\ref{sec:trails}, we introduce \emph{trails} as a data-independent abstraction 
of such sequences, and show that it gives rise to a rule precedence with desirable properties.
Since recognising trails is undecidable, we define \emph{chains} as a decidable relaxation (Section~\ref{sec:chains}),
which allows us to introduce \emph{chain stratification} and to show that it meets our requirements (Section~\ref{sec:chain-strat}).

Deciding this new form of stratification remains challenging, since chains can be arbitrarily long.
By relating chains to words in a regular language, we finally obtain a decision procedure for chain stratification
(Section~\ref{sec:reg-lang}).
\ifTechnicalReport
To show its practical feasibility, Section~\ref{sec:impl} presents a prototype implementation that builds upon algorithmic methods
for reliances \cite{GIKM2022}, and an evaluation that shows applicability to rule sets with thousands of rules.
Additional details and proofs can be found in the appendix, and our source code is included in the supplementary material statement at the end of this paper.
\fi

 \section{Preliminaries}\label{sec:prelim}

We briefly define rules and recall required notions about reliances \cite{KR20:cores}.
We use first-order notation with predicates (not just \emph{triple}) --- our results
apply to N3 (under common syntax translations \cite{existentialN3}) but also to other rule
languages \cite{RIF-overview,Ivliev+:Nemo2024,BSG:Vadalog18}.

\myparagraph{Rules}
We use countably infinite, mutually disjoint sets $\Pred$ (\emph{predicates}), 
$\V$ (\emph{variables}), $\C$ (\emph{constants}), and $\N$ (\emph{nulls}).
Each predicate $p\in\Pred$ has an arity $\ar(p)\geq 0$.
An \emph{atom} is an expression $p(\vec{t})$ with $p\in\Pred$ and $\vec{t}\in(\V\cup\C\cup\N)^{\ar(p)}$ a list.
A \emph{fact} is a variable-free atom, and a \emph{database} is a set $\I$ of facts.
Negation is denoted $\neg$. We often treat lists and conjunctions of (negated) atoms as sets.
For any expression or set of expressions $E$, $\var(E)$ is the set of variables in $E$.

A \emph{rule} $\rho$ is a null-free expression $\rho: \bodyP(\rho)\land \bodyN(\rho) \ruleTo \exists\vec{z}.\,\head(\rho)$,
where $\vec{z}$ is a variable list,
$\bodyP(\rho)$ and $\head(\rho)$ are conjunctions of atoms, and $\bodyN(\rho)$ is a conjunction of negated atoms.
Variables $\varE(\rho)=\vec{z}$ are \emph{existential}, all others $\varA(\rho)=\var(\rho)\setminus\vec{z}$
are \emph{universal} (as usual, we omit their quantifiers).
Universal variables in $\head(\rho)$ are \emph{frontier variables}.
All universal variables in $\head(\rho)$ and $\bodyN(\rho)$ must occur in $\bodyP(\rho)$ (\emph{safety}).
A rule $\rho$ is \emph{Datalog} if $\bodyN(\rho)=\varnothing$ and $\varE(\rho)=\varnothing$.
A \emph{ruleset} $\R$ is a finite set of rules; then $\RpD\subseteq\R$ is the subset of Datalog rules. 
Constraints as in \alglineref{rule_constraint_teacher_participant}--\alglineref{rule_constraint_teacher_type} are
encoded in Datalog using a nullary head predicate $\bot$. 
We will discard cases where $\bot$ would be derived in our analyses, but not otherwise endow it with special logical semantics.

\myparagraph{Pieces}
A \emph{piece} $\psi$ of $\head(\rho)$ is a minimal non-empty set $\psi\subseteq\head(\rho)$
such that $\varE(\psi)$ is disjoint from $\varE(\head(\rho)\setminus\psi)$.
Let $\pieces(\rho) = \{ \exists \varE(\psi).\psi \mid \psi\text{ a piece of }\head(\rho)\}$, and
$\pieces(\R) = \bigcup_{\rho\in\R}\pieces(\rho)$.
$\R$ is \emph{piece-decomposed} if each $\rho\in\R$ has just one piece.
\begin{example}
    The rule $a(x) \ruleTo \exists v,w.\,b(x),r(x,v),b(v),s(x,w)$ has three pieces: $b(x)$, $\exists v.\,r(x,v),b(v)$ and $\exists w.\,s(x,w)$.
\end{example}

\myparagraph{Matches and models}
A \emph{substitution} is a partial mapping $\sigma:\V\to(\V\cup\C\cup\N)$.
For $t\in\V\cup\C\cup\N$, we set $t\sigma=\sigma(t)$ if $\sigma(t)$ is defined, and $t\sigma=t$ otherwise.
Substitutions extend to (sets of) logical expressions as usual.
Rule $\rho$ \emph{matches} database $\I$ if there is a substitution $\mu_\forall$ defined on $\varA(\rho)$
such that $\bodyP(\rho)\mu_\forall\subseteq\I$ and $\bodyN(\rho)\mu_\forall\cap\I=\varnothing$.
We denote matches as pairs $\tuple{\rho,\mu_\forall}$.
Match $\tuple{\rho,\mu_\forall}$ is \emph{satisfied} by $\I$ if there is a substitution $\mu_\exists$ defined on $\varE(\rho)$
such that $\head(\rho)\mu_\forall\mu_\exists\subseteq\I$.
$\I$ is a \emph{model} of $\rho$ (or \emph{satisfies} $\rho$, written $\I\models\rho$) if it
satisfies all matches of $\rho$ over $\I$.
$\I$ \emph{models} a ruleset $\R$ if $\I\models\rho$ for all $\rho\in\R$, and $\I$ \emph{models} a null-free database $\mathcal{J}$ if
$\mathcal{J}\subseteq\I$.

\myparagraph{The Chase}
Algorithms that construct models for a given database and ruleset are called \emph{chase procedures}.
We use the following \emph{standard chase}:
\begin{definition}\label{def:chase}
    A \emph{chase} $C$ for a null-free database $\I$ and ruleset $\R$ 
    is a (possibly infinite) sequence of databases $\cDB(C, 0), \cDB(C, 1), \ldots$ such that 
    $\cDB(C, 0) = \I$ and:
    \begin{enumerate}
\item[(1)]\label{it_chase_apply} for each $i>0$, there is $\cRule(C,i)\in\R$ and $\cMatch(C,i)=\tuple{\cRule(C,i),\mu_\forall}$, such that 
            (a) $\cMatch(C,i)$ is unsatisfied in $\cDB(C,i-1)$, and
            (b) $\cDB(C,i)=\cDB(C,i-1)\cup\{ \head(\cRule(C,i))\mu_\forall\mu_\exists \}$ for an injective function
            $\mu_\exists: \varE(\cRule(C,i))\injTo\N$ that maps variables to distinct fresh nulls (not occurring in $\cDB(C,i-1)$);
        \item[(2)] if $\tuple{\rho,\mu}$ is a match over $\cDB(C,i)$ for some $i\geq 0$, then there is $j\geq i$ 
            such that $\tuple{\rho,\mu}$ is satisfied in $\cDB(C,j)$ (fairness).
    \end{enumerate} 
    The set of \emph{chase steps} is $\cSteps(C) \subseteq \mathbb{N}$.
The result of a chase is $\chase_C(\I,\R) = \bigcup_{i\in\cSteps(C)} \cDB(C,i)$.
    $C$ is \emph{generating} if every $\cMatch(C,i)$ is a match in $\chase_C(\I,\R)$.
\end{definition}
In situations like \eqref{it_chase_apply}, we say \emph{$\cDB(C,i)$ was obtained by applying $\cRule(C,i)$ for $\mu_\forall\mu_\exists$}. 
Many fair chase sequences may exist, based on the choice of rule applied in each step.
If $C$ is generating, negated bodies of applied rules remain satisfied. If $\R$ has no existentials,
the results of generating chases are exactly the \emph{stable models} \cite{KLS6:coloring-asp}.
If $\R$ is Datalog, $\chase_C(\I,\R)$ is the unique \emph{perfect model},
and we omit ${}_C$.

Likewise, the chase result is the unique perfect model if $\R$ is \emph{stratified}:
Let $\opfont{preds}(A)$ be the set of predicates used in the set of atoms $A$,
and for $\circ \in \{ +,- \}$, let $\prec^\circ_{\opfont{preds}} \coloneqq \{ \langle\rho_1,\rho_2\rangle\in\R^2 \mid \opfont{preds}(\head(\rho_1)) \cap \opfont{preds}(\opfont{body}^\circ(\rho_2)) \neq\varnothing \}$.
Then an existential-free ruleset $\R$ is called (classically) stratified if
it can be partitioned into $\R = S_0 \mathrel{\dot{\cup}} \ldots \mathrel{\dot{\cup}} S_n$
such that for all $\rho_i\in S_i$ and $\rho_j\in S_j$ we have that
$\rho_i \prec^+_{\opfont{preds}} \rho_j$ implies $i \leq j$, and
$\rho_i \prec^-_{\opfont{preds}} \rho_j$ implies $i < j$.
Typically, it is also required that the sets of predicates $\bigcup_{\rho\in S_i} \opfont{preds}(\head(\rho))$ defined by the strata $S_i$ are disjoint, which is inessential.
$\R$ is stratified iff $\rtHull{(\prec^+_{\opfont{preds}})}\circ\prec^-_{\opfont{preds}}$ is acyclic.
A chase exhaustively applying rules from the strata in order is clearly generating.
\par\nobreak
Things are not so simple with $\exists$ (cf.\ Section~\ref{sec:intro}, Problem~2).

\myparagraph{Universal models and cores}
Models with nulls are compared using \emph{homomorphisms} (see \cite{KR20:cores} for a definition). 
A model $\mathcal{U}$ of $\R$ and $\I$ is \emph{universal} if all other models $\mathcal{M}$ of
$\R$ and $\I$ have a homomorphism to $\mathcal{U}$. Universality is the basis for answering monotonic queries \cite{DNR08:corechase}.
A \emph{core} is a database without ``redundant nulls'', which can be defined on finite databases $\I$  
by requiring that every homomorphism $\I\to\I$ is surjective (the infinite case is more subtle \cite{KR20:cores}).
The ``minimality'' of cores makes them useful for non-monotonic queries.
A local criterion for a chase to yield a core is based on \emph{alternative matches}:
\begin{definition}\label{def:altM}
    Let $\I_a \subseteq \I_b$ be databases, such that $\I_a$ was obtained by applying $\rho$ for $\mu$. 
    A mapping $\mu^A:\var(\rho)\to\C\cup\N$ is an \emph{alternative match} for $\tuple{\rho,\mu}$ over $\I_b$
    if
    (1) $\head(\rho)\mu^A \subseteq\I_b$,
    (2) $x\mu=x\mu^A$ for all $x\in\varA(\rho)$, and
    (3) there is a null in $\head(\rho)\mu$ that is not in $\head(\rho)\mu^A$. 
An alternative match in a chase $C$ is one that uses $\I_a=\cDB(C,i)$, $\rho=\cRule(C,i)$, and $\I_b=\chase_C(\I,\R)$.
\end{definition}
If a generating chase is finite and has no alternative matches, then it yields a core,
but it is undecidable if such a chase exists \cite[Thms~4~\&~11]{KR20:cores}.
Moreover, rules \alglineref{rule_father}--\alglineref{rule_different_fathers} have
generating chases $C$ and $C'$ where all databases $\cDB(C^{(\prime)},i)$ are cores, but with non-homomorphic (hence non-universal) results.
Decidable criteria for unique universal core models can be defined with rule precedences.

\myparagraph{Precedences}
A \emph{precedence} for ruleset $\R$ is a strict partial order ${\prec}\subseteq\R\times\R$.
A chase $C$ \emph{violates} $\prec$ if $\cRule(C,i)\prec\cRule(C,j)$ for some $i>j$.
$C$ \emph{respects} $\prec$ if, for every step $i$,
$\cDB(C,i)$ has no unsatisfied match $\tuple{\rho,\mu}$ with $\rho\prec\cRule(C,i)$.
This is different from ``non-violating'' since unsatisfied $\tuple{\rho,\mu}$ may become satisfied
without $\rho$ being applied.
Krötzsch defines precedences based on three types of \emph{reliance} relations, which were shown
to be efficiently computable \cite{GIKM2022}.

\begin{definition}\label{def:posr}
    Rule $\rho_2$ \emph{positively relies on} rule $\rho_1$, written $\rho_1\posr\rho_2$, if there are databases
    $\I_a\subseteq \I_b$ such that $\I_b$ was obtained from $\I_a$ by applying $\rho_1$ for $\mu_1$,
    and there is an unsatisfied match $\tuple{\rho_2,\mu_2}$ over $\I_b$ that is not a match over $\I_a$.
\end{definition}

\begin{definition}\label{def:negr}
    Rule $\rho_2$ \emph{negatively relies on} rule $\rho_1$, written $\rho_1\negr\rho_2$, if there are databases $\I_a\subseteq \I_b$
    such that $\I_b$ was obtained from $\I_a$ by applying $\rho_1$ for $\mu_1$,
    and there is an unsatisfied match $\tuple{\rho_2,\mu_2}$ over $\I_a$ that is not a match over $\I_b$.
\end{definition}

\begin{definition}\label{def:restr}
    Rule $\rho_1$ \emph{restrains} rule $\rho_2$, written $\rho_1\restr\rho_2$, if there are databases $\I_a\subseteq \I_b$
    such that
    (1) $\I_a$ was obtained by applying $\rho_2$ for $\mu_2$,
    (2) $\I_b$ was obtained by applying $\rho_1$ for $\mu_1$,
    and (3) there is an alternative match $\mu^A$ for $\tuple{\rho_2,\mu_2}$ over $\I_b$
    that is not an alternative match over $\I_b\setminus\head(\rho_1)\mu_1$.
\end{definition}

The intuition for these notions is that $\rho_1$ may produce an inference
that is needed to match $\rho_2$ ($\posr$), 
prevents matching $\rho_2$ ($\negr$), or
creates an alternative match for $\rho_2$ ($\restr$).
A chase that does not violate $\negr$ is generating.
A chase that does not violate $\restr$ has no alternative matches.

A ruleset $\R$ is \emph{core stratified} if the relation ${\posr}\cup{\restr}$ has no cycles through
$\restr$, and \emph{fully stratified} if ${\posr}\cup{\negr}\cup{\restr}$ has no cycles through
${\negr}\cup{\restr}$. For cases where a finite model exists, full stratification (which respects $\rtHull{(\posr)}\circ(\negr\cup\restr)$) yields a generating chase sequence
that (if terminating) results in a uniquely determined core model, called the \emph{perfect core model} \cite{KR20:cores}.
The conditions can be relaxed if negation is only used for some predicates \cite{EKM:ExNeg2022}.

 \section{Trails: Refining Transitive Positive Reliances}\label{sec:trails}

To determine a precedence that yields a perfect core model, it is necessary to estimate when the application of one rule may cause a new unsatisfied match for another rule in the future.
This may happen directly, as approximated with $\posr$, or possibly involve multiple intermediate steps. Previously, this was achieved by considering the reflexive transitive closure $\rtHull{(\posr)}$ of positive reliances \cite{KR20:cores},
which is a coarse overestimation.
A $\posr$-path does not detect whether the sequence of rule applications leading to a match
implies that this match is always satisfied.
This section aims to improve this by identifying a proper sub-relation of $\rtHull{(\posr)}$ that records
the nuances of such multistep interactions.

To represent sequences of matches, we will use words over an infinite set of \emph{rule instances} (individually denoted by $\ins$):
\begin{equation}
    \instances(\R) \coloneqq \{ \rho\theta \mid \rho\in\R, \theta: \var(\rho)\to (\V \setminus \var(\R)) \}
\end{equation}

\begin{example}\label{ex:instances}
    Below, rules $\rho_1,\rho_2$ and $\rho_3$ on the left are shown with one of their (infinitely many) instances on the right:
    \begin{equation*}\begin{array}{l|l}
        \rho_1\!: t(x,y,z), q(x,y) \ruleTo p(x,z)              \:&\: \ins_1\!: t(x_1,y_1,x_1), q(x_1,y_1) \ruleTo p(x_1,x_1) \\
        \rho_2\!: p(x,x) \ruleTo \exists v.\, s(x,v), t(v,v,v) \:&\: \ins_2\!: p(x_1,x_1) \ruleTo \exists v_1.\, s(x_1,v_1), t(v_1,v_1,v_1)  \\
        \rho_3\!: q(x,y), s(x,z), a(z) \ruleTo t(x,x,z)        \:&\: \ins_3\!: q(x_1,y_1), s(x_1,v_1), a(v_1) \ruleTo t(x_1,x_1,v_1)
    \end{array}\end{equation*}
    The rule instance $\ins_1$ was obtained from rule $\rho_1$ by applying the variable substitution
    $x \xmapsto{\theta} x_1$, $y \xmapsto{\theta} y_1$, $z \xmapsto{\theta} x_1$, collapsing variables $x$ and $z$.

    We will reuse these rules as a running example throughout Sections~\ref{sec:trails}--\ref{sec:chain-strat}.
\end{example}

Together with a fixed arbitrary bijection $\omega: (\V \setminus \var(\R))\injTo \C\cup\N$, each instance represents a match.
Applying $\omega$ to a set of atoms produces a database and, conversely, applying $\omega^{-1}$ to a database yields a set of atoms.
A sequence of matches $\langle\rho_1,\mu_1\rangle, \ldots, \langle\rho_n,\mu_n\rangle$ is now represented by
a word $\rho_1\mu_1\omega^{-1} \ldots \rho_n \mu_n \omega^{-1}$ over alphabet $\instances(\R)$.
As the chase algorithm always maps existentials to fresh nulls,
we require these words to use fresh variables: 

\begin{definition}\label{def:exists-disjoint}
    A finite sequence $(\ins_i)_{i=1,\ldots,n} \in \instances(\R)^*$ of rule instances $\ins_i \in \instances(\R)$
    with pairwise disjoint existential variables
is called \emph{$\exists$-disjoint}.
\end{definition}

\begin{example}[\cont{ex:instances}]\label{ex:exists-disjoint}
    The sequence of instances $\ins_1 \ins_2 \ins_3$ is $\exists$-disjoint, as the only existential variable $v_1$ is fresh in $\ins_2$.
    (It can be used universally again, as in $\ins_3$.)
\end{example}

Such a word indicates that the rule application corresponding to its first instance
may indirectly cause a new unsatisfied match for the rule whose instance ends the word,
if one can give a witnessing chase which applies the matches in the given order,
such that each match directly requires the previous one.

\begin{definition}\label{def:trail}
    An $\exists$-disjoint sequence of rule instances $t=\ins_1,\ldots,\ins_n$ is a \emph{trail} if
    there exists a generating chase sequence $C$,
    for which there is a strictly increasing function
$f: \set[n] \injTo \cSteps(C)$,
    such that for all $i\in\set[n]$
    \begin{itemize}
        \item $\langle \ins_i,\omega \rangle = \cMatch(C,f(i))$ and \hfill (match correspondence)
        \item $i>1$ implies $\bodyP(\ins_i)\omega \cap \head(\ins_{i-1})\omega \neq \varnothing$. \hfill (causal connection)
    \end{itemize}
    We write $\rho_1\posrT \rho_2$ if there is a trail from an instance of $\rho_1$ to an instance of $\rho_2$.
\end{definition}

Clearly, we have ${\posrT} \subseteq {\rtHull{(\posr)}}$, as the databases $\I_a=\cDB(C,f(i)-1)$ and $\I_b=\cDB(C,f(i))$ witness the pairwise positive reliances for all $i$.

\begin{example}[\cont{ex:exists-disjoint}]\label{ex:trail}
    Let $c_{x_1},c_{y_1} \in \C$, $n_{v_1} \in \N$, $x_1 \xmapsto{\omega} c_{x_1}$, $y_1 \xmapsto{\omega} c_{y_1}$, $v_1 \xmapsto{\omega} n_{v_1}$.
    The $\exists$-disjoint sequence $\ins_1 \ins_2 \ins_3$ is a trail,
    as there is a chase $C$ with
\begin{equation*}
        \begin{array}{l}
            \cDB(C,0) = \{ t(c_{x_1},c_{y_1},c_{x_1}), q(c_{x_1},c_{y_1}) \} \hfill
            \cDB(C,1) = \cDB(C,0) \cup \{ p(c_{x_1}, c_{x_1}) \} \\
            \cDB(C,2) = \cDB(C,1) \cup \{ s(c_{x_1},n_{v_1}), t(n_{v_1},n_{v_1},n_{v_1}) \} \quad\>\>
            \cDB(C,3) = \cDB(C,2) \cup \{ a(n_{v_1}) \} \\
            \cDB(C,4) = \cDB(C,3) \cup \{ t(c_{x_1},c_{x_1},n_{v_1}) \}
        \end{array}
    \end{equation*}
    and step mapping $f: \{1,2,3\} \injTo \{1,2,3,4\}$ with $1\mapsto 1$, $2\mapsto 2$, $3\mapsto 4$,
    assuming a rule, e.g. $s(x,y)\ruleTo a(y)$, to derive fact $a(n_{v_1})$ in step $3$.
    Hence, $\rho_1 \posrT \rho_3$.
\end{example}

\begin{restatable}{theorem}{thmPosrTPrecNoViolate}\label{thm:posrT-prec-no-violate}
If ${\posrT} \circ (\negr\cup\restr)$ is a precedence (hence acyclic),
then a chase that respects it does not violate it. In particular, it does not violate $\negr$ or $\restr$,
and therefore is generating and free of alternative matches.
\end{restatable}

Theorem~\ref{thm:posrT-prec-no-violate} can be shown by a contrapositive argument, where we take a chase that violates ${\posrT} \circ {(\negr\cup\restr)}$,
inspect the first violation at chase step $i$ and argue that the chase must disrespect the precedence at $i$.
The other properties follow from slight adaptations of known results on $\negr$ and $\restr$ \cite{KR20:cores}.

\begin{restatable}{theorem}{thmPosrTPrecUniqeChaseResult}\label{thm:posrT-prec-unique-chase-result}
If ${\posrT} \circ (\negr \cup \restr)$ is a precedence, then the
result of every chase of $\I$ and $\R$ that respects it is unique (up to isomorphism). \end{restatable}

Uniqueness is shown by arguing that of two non-isomorphic chases, one must disrespect the precedence.
A chase as in Theorem~\ref{thm:posrT-prec-unique-chase-result} is generating and free of alternative matches by Theorem~\ref{thm:posrT-prec-no-violate}.
If finite, this unique result is a core \cite[Theorem~11]{KR20:cores}, known as
the \emph{perfect core model} of $\I$ and $\R$.

\begin{restatable}{lemma}{lemPosrTSmallest}\label{lem:posrT-smallest}
    The $\posrT$ relation is the smallest relation $<$,
    such that respecting $< \circ \prec$ ensures non-violation of $< \circ \prec$
    for acyclic $< \circ \prec$ with ${\negr} \subseteq {\prec}$.
\end{restatable}

Based on proof of Theorem~\ref{thm:posrT-prec-no-violate} outlined above, this is shown by contradiction.
Lemma~\ref{lem:posrT-smallest} therefore confirms that our choice of $\posrT$ is optimal for our purposes.

\begin{restatable}{theorem}{thmTrailUndecidable}\label{thm:trail-undecidable}
    For $\ins_1,\ins_2\in\instances(\R)$, it is undecidable whether $\ins_1\ins_2$ is a trail.
\end{restatable}

Theorem~\ref{thm:trail-undecidable} is unsurprising, as chase termination is undecidable.
Simulating the run of a Turing-machine with existential rules is a known technique
that can be adapted such that the ruleset has a trail iff the TM halts.

\begin{remark}
    As a direct consequence of Theorem~\ref{thm:trail-undecidable},
    it is also undecidable whether there is any trail starting and ending with instances of given rules,
    since that is already the case for only two instances.
\end{remark}

Consequently, we are looking for a decidable relation between $\posrT$ and $\rtHull{(\posr)}$.  \section{Chains: A Decidable Approximation}\label{sec:chains}

By Section~\ref{sec:trails}, $\posrT$ is the smallest sub-relation of $\rtHull{(\posr)}$ that sufficiently captures when
one rule application may enable another.
We now develop a decidable overestimation of $\posrT$ based on positive reliance checks.
When checking $\posr$, it suffices to consider instances under $\omega$, which
leads to $\iposr$ as a special case:

\begin{definition}\label{def:direct-posr}
	For $\ins_1, \ins_2 \in \instances(\R)$, we write
	$\ins_1\iposr\ins_2$, if
	$\ins_1\posr\ins_2$ with
	$\I_a=(\bodyP(\ins_1)\cup(\bodyP(\ins_2)\setminus\head(\ins_1)))\omega$,
and
	$\mu_1=\mu_2=\omega$.
\end{definition}
Above, database $\I_a$ and matches $\mu_1$ and $\mu_2$ refer to Definition~\ref{def:posr}.

\begin{example}[\cont{ex:trail}]\label{ex:direct-posr}
    It is clear that $\ins_1\iposr\ins_2$, as $\langle\ins_2,\omega\rangle$ is no match on database
    $\I_a = \{ t(c_{x_1}, c_{y_1}, c_{x_1}), q(c_{x_1}, c_{y_1}) \}$,
    but applying the unsatisfied match $\langle\ins_1,\omega\rangle$
    yields $\I_b = \I_a \cup \{ p(c_{x_1}, c_{x_1}) \}$
    where $\langle\ins_2,\omega\rangle$ is an unsatisfied match.
\end{example}

It is natural to ask if a sequence $c$ of instances --- as a whole --- can cause a new unsatisfied match $\langle\rho,\mu\rangle$,
i.e., if it can be \emph{extended} by an instance $\ins \coloneqq \rho \mu \omega^{-1}$.
To determine this, we define a \emph{chain rule} $\rho_c$ that is obtained from $c$
by combining all body atoms from the instances as well as all head atoms from all but the last instance $\ins_k$ into $\bodyP(\rho_c)$
and keeping the head of $\ins_k$ as $\head(\rho_c)$.
Any match for $\ins_k$ that arose due to prior applications of instances $\ins_1 \ldots \ins_{k-1}$
also satisfies the combined list of pre-conditions in $\bodyP(\rho_c)$.

\begin{definition}\label{def:extend}
An $\exists$-disjoint sequence of rule instances $c=\ins_1,\ldots,\ins_k$ can be \emph{extended by $\ins\in\instances(\R)$}
    if $\rho_c\iposr\ins$, where $\rho_c$ is the \emph{chain rule}
    \begin{align}\label{eq:chain-rule}\left(\textstyle\bigcup_{i\in\set[k]} \bodyP(\ins_i)\right)\cup
            \left(\textstyle\bigcup_{i\in\set[k-1]} \head(\ins_i)\right) \ruleTo \exists \varE(\ins_k) .\, \head(\ins_k)
\end{align}
\end{definition}
\begin{definition}\label{def:chain}
An $\exists$-disjoint sequence of rule instances $c=\ins_1,\ldots,\ins_k$ is a \emph{chain} if
$k=1$ or, recursively, $\ins_1\ldots\ins_{k-1}$ is a chain that can be extended by $\ins_k$.
\end{definition}

\begin{example}[\cont{ex:direct-posr}]\label{ex:chain}
    To see that $c = \ins_1 \ins_2 \ins_3$ is a chain, check that
    $\ins_1$ is a chain of length $1$,
    $\rho_{\ins_1} \iposr \ins_2$ (see Example~\ref{ex:direct-posr}), and
    $\rho_{\ins_1 \ins_2} \iposr \ins_3$ where the chain rule for $\ins_1\ins_2$ is
    $\rho_{\ins_1 \ins_2} : t(x_1,y_1,x_1), q(x_1,y_1), p(x_1,x_1) \ruleTo \exists v_1.\, s(x_1,v_1), t(v_1,v_1,v_1)$.
\end{example}

\begin{remark}\label{obs:trail-is-chain}
    Every trail is a chain.
\end{remark}

Although it is easy to determine whether a given sequence of rule instances forms a chain,
testing whether there is \emph{any} chain
between the instances of two rules
is difficult, as chains can in principle grow arbitrarily long.
We can, however, further restrict to certain chains:

\begin{definition}\label{def:decoupled}
    For a sequence $s=\ins_1,\ldots,\ins_n$,
    the set of \emph{stale variables} is:
    \begin{equation}
        \stale(s) \coloneqq \textstyle\bigcup_{i,j\in\set[n]. j<i} \left(\left( \varA(\ins_i) \setminus \var(\head(\ins_{i-1})) \right) \cap \varA(\ins_j)\right)
    \end{equation}
    If $\stale(s)=\varnothing$, then $s$ is \emph{decoupled}.
\end{definition}
Decoupled chains are special because the only variables re-used in the $i$th instance are frontier and existential variables stemming from the $(i-1)$th instance.

\begin{example}[\cont{ex:chain}]\label{ex:decoupled}
    The sequence $t=\ins_1 \ins_2 \ins_3$ is not decoupled,
    as $y_1$ is re-used in $\ins_3$ although not present in $\var(\head(\ins_2))=\{ x_1, v_1 \}$
    and thus $\stale(t)=\{ y_1 \} \neq \varnothing$.
\end{example}

This shows that trails are not always decoupled. We can, however, relate each trail to a decoupled chain.
\begin{definition}\label{def:generalised}
    A sequence $\ins^A_1,\ldots,\ins^A_n$ \emph{generalises} a sequence $\ins^B_1,\ldots,\ins^B_n$
    if there is a variable substitution $\sigma: \V\to\V$, such that $\ins^A_i\sigma = \ins^B_i$ for all $i\in\set[n]$.
\end{definition}

\begin{example}[\cont{ex:decoupled}]\label{ex:generalised}
    Any sequence of rule instances can be rewritten such that it is decoupled, by injectively replacing the stale variables with fresh ones.
    Consider $\theta_1 = \theta_2 = \id$ and $\theta_3$ with $y_1 \mapsto y_2$.
    Then $c = (\ins_1\theta_1) (\ins_2\theta_2) (\ins_3\theta_3)$ is decoupled,
    and $c$ generalises $t$ with $\sigma = \theta_1^{-1} \circ \theta_2^{-1} \circ \theta_3^{-1}$, which simply replaces $y_2$ by $y_1$.
\end{example}
\begin{restatable}{lemma}{lemTrailGenDecChain}\label{lem:trail-gen-dec-chain}
    Every trail can be generalised by a decoupled chain.
\end{restatable}

Lemma~\ref{lem:trail-gen-dec-chain} is shown by naming the trail's stale variables apart,
along the lines of Example~\ref{ex:generalised}.
The resulting word may not be a trail any more, but is necessarily a chain.
To show $\rho_{\ins_1\theta_1\ldots\ins_{k-1}\theta_{k-1}}\iposr\ins_k\theta_k$ for all word positions $k$,
we consider $\I_a$ and $\I_b$ obtained from the original trail's witnessing chase and augmented
with additional copies of facts under suitable renamings $\omega^{-1}\theta_1\omega$ up to $\omega^{-1}\theta_k\omega$.

\begin{example}[\cont{ex:generalised}]\label{ex:trail-gen-dec-chain}
    Database $\I_a$ from Definition~\ref{def:direct-posr} will, in addition to $q(c_{x_1},c_{y_1})$,
    contain a fact $q(c_{x_1},c_{y_2})$ matching $\bodyP(\ins_3\theta_3)$ with named apart variable $y_2$.
\end{example}

For ruleset $\R$, we collect all decoupled chains in set $\chains(\R)\subseteq \instances(\R)^*$, which may be infinite.
We write $\rho_1 \posrC \rho_2$ if there is a decoupled chain between instances of $\rho_1$ and $\rho_2$.
With Lemma~\ref{lem:trail-gen-dec-chain}, we know that indeed ${\posrT} \subseteq {\posrC}$, hence:

\begin{corollary}[of Theorem~\ref{thm:posrT-prec-no-violate}]\label{cor:posrTc-restr-altM-free}
    A chase that respects ${\posrC}\circ(\negr\cup\restr)$ does not violate $\negr$ or $\restr$, and therefore
is generating and free of alternative matches.
\end{corollary}

However, as the following example shows, there are decoupled chains that do not generalise any trail, i.e. ${\posrT} \subset {\posrC}$ (the containment is proper).

\begin{example}[\cont{ex:trail-gen-dec-chain}]\label{ex:dec-chain-not-gen-trail}
    Recall that---in order to give the witnessing chase for the trail $t = \ins_1 \ins_2 \ins_3$ in Example~\ref{ex:trail}---we had to assume the existence of a rule to derive a fact $a(n_{v_1})$ after applying instance $\ins_2$, which invented the null $n_{v_1}$.
    This fact cannot originate from the input database, as it pertains to a null.
    If, however, the full ruleset were $\R = \{ \rho_1,\rho_2,\rho_3 \}$, $t$ would not be a trail (but still a chain).
\end{example}

If $c=\ins_1,\ldots,\ins_n$ is a chain, then $\ins_i\posr\ins_{i+1}$ for $i\in\set[n-1]$, but the reverse does not hold,
    as attested below:
The inclusion ${\posrC} \subset {\rtHull{(\posr)}}$ is proper. 

\begin{example}\label{ex:dl-posr-seq-not-chain}
    Consider the rules $\rho_1 : r(x,y), a(y) \ruleTo b(x)$, $\rho_2 : b(x) \ruleTo c(x)$, and $\rho_3 : c(x) \ruleTo \exists v.\, r(x,v), a(v)$,
    which represent common types of ontological axioms.
They are pairwise positively relying, i.e. $\rho_1\posr\rho_2\posr\rho_3$, but there are no instances thereof that would form a chain,
    as the way $\rho_1$ derives $b(x)$ ensures that $\head(\rho_3)$ is already satisfied for any match of $\rho_3$ introduced by applying $\rho_1$ and then $\rho_2$.
    Therefore, the relation $\posrC$ is indeed a proper sub-relation of $\rtHull{(\posr)}$.
\end{example}

In Section~\ref{sec:chain-strat} we further refine the composition of $\posrC$ with $\restr$ and $\negr$, and use it to introduce chain stratification.
Afterwards, in Section~\ref{sec:reg-lang}, we examine how $\posrC$ can be computed despite the fact that chains could be arbitrarily long.
 \section{Rule Selection for Chain-Stratified Rulesets}\label{sec:chain-strat}

Having established that respecting ${\posrC} \circ {(\negr\cup\restr)}$ prevents
violation of $\negr$ and $\restr$,
we now take a closer look at the negative or restraint reliance after the chain's last instance:
We introduce relations ${\negrC} \subset {\posrC\circ\negr}$ and ${\restrC} \subset {\posrC\circ\restr}$,
and use them to demarcate a class of rulesets which we call (fully) chain-stratified.
These subsume fully stratified rulesets \cite{KR20:cores}.

\begin{definition}\label{def:restrC}
    A chain $c = \ins_1\ldots\ins_n$ restrains a rule $\rho$ if $\rho_c \restr \rho$ where $\rho_c$ is the chain rule of $c$.
    If $\ins_1$ is an instance of $\rho_1\in\R$, then we write $\rho_1 \restrC \rho$.
\end{definition}

Clearly, $\rho_1 \restrC \rho$ implies $\rho_1 \posrC \circ \restr \rho$,
but the reverse does not hold:

\begin{example}[\contof{ex:dec-chain-not-gen-trail}]\label{ex:restrC}
    The last instance $\ins_3$ restrains both rules
$\rho_4: p(x,x) \ruleTo \exists v. t(v,x,v)$
    and $\rho_5: q(x,y) \ruleTo \exists v.\, t(x,y,v)$ but $\rho_c$ only restrains $\rho_5$. \end{example}

Respecting $\restrC$ still prevents $({\posrC}\circ{\restr})$-violations causing alternative matches.

\begin{restatable}{lemma}{lemTrailRestrC}\label{lem:trail-restrC}
    Let $t=\ins_1\ldots\ins_n$ be a trail and $C$ a witnessing chase for $t$ with step mapping $f$.
    If there is a chase step $i < f(n)$, such that $\cRule(C,i)$ has an alternative match over $\cDB(C,f(n))$ but not over $\cDB(C,f(n)-1)$,
    then $\ins_1 \restrC \cRule(C,i)$.
\end{restatable}

This is shown by obtaining a pair of databases from the offending chase and applying Definition~\ref{def:restr}. Negative reliances are treated similarly:

\begin{definition}\label{def:negrC}
    A rule $\rho$ negatively relies on a chain $c = \ins_1\ldots\ins_n$ if $\rho_c \negr \rho$.
    If $\ins_1$ is an instance of $\rho_1\in\R$, then we write $\rho_1 \negrC \rho$.
\end{definition}\begin{restatable}{lemma}{lemTrailNegrC}\label{lem:trail-negrC}
    Let $t=\ins_1\ldots\ins_n$ be a trail and $C$ a witnessing chase for $t$ with step mapping $f$.
    If there is a chase step $i < f(n)$ such that $\cMatch(C,i)$ is a match for $\cRule(C,i)$ over $\cDB(C,f(n)-1)$ but not over $\cDB(C,f(n))$,
    then $\ins_1 \negrC \cRule(C,i)$.
\end{restatable}

With the above improvements, we can finally define our novel stratification condition, which we call \emph{chain stratification}:

\begin{definition}\label{def:chain-stratified}
    A ruleset $\R$ is \emph{core chain-stratified} if $\langle \R, \restrC \rangle$ is acyclic.
    It is \emph{(fully) chain-stratified} if $\langle \R, \negrC \cup \restrC \rangle$ is acyclic.
\end{definition}

\begin{corollary}[of Theorems~\ref{thm:posrT-prec-no-violate} and \ref{thm:posrT-prec-unique-chase-result}, and Lemma~\ref{lem:trail-gen-dec-chain}]\label{cor:core-chain-stratified}
On a chain-stratified ruleset the result of any chase that respects $\negrC \cup \restrC$ is generating, alternative match-free, and unique (up to isomorphism).
 If finite, it is a perfect core model.
\end{corollary}

For chain-stratified rulesets, the rule precedence $\negrC\cup\restrC$ induces a ``stratification'' where rules may belong to multiple strata,
obtained by first computing the minimum-rank layering $L_0\ldots L_n$ of the precedence,
and then forming strata $S_i = \bigcup_{j\leq i} L_j$. Successively chasing fixpoints for $S_0\ldots S_n$ ensures that no new unsatisfied matches for any $(\negrC\cup\restrC)$-predecessors of $S_i$ are added.
 \section{The Regular Language of Chains}\label{sec:reg-lang}

We now address the issue of deciding whether there is an (arbitrarily long) decoupled chain connecting instances of two given rules.
Enumerating the set $\chains(\R)$
merely yields a semi-decision procedure,
because a ruleset can have infinitely many chains.
Therefore, we characterise $\chains(\R)$ with a language $\L(\R)$ and establish that $\L(\R)$ is regular,
which yields a decision procedure.
Our restriction to \emph{decoupled} chains (Section~\ref{sec:chains}) is essential for this to work.

We assign to each $c\in\chains(\R)$ a label $\ell(c)\in \Alphabet$ from some (yet to be defined) alphabet $\Alphabet$.
Intuitively, $\ell(c)$ captures all information needed for reliance checks between the
chain rule $\rho_c$ and any rule from $\R$, but in bounded space that does not grow arbitrarily with (the body of) $\rho_c$.
This requires some auxiliary notation.
For ruleset $\R$, the finite set
$\variants(\R) \coloneqq \{ \rho\theta \mid \rho\in\R, \theta: \var(\rho)\to \var(\R) \}$\phantomsection\label{eq:variants}
is obtained by replacing variables in $\R$ by any combination of variables found in $\R$.
Further, let $\orig: \instances(\R) \to \variants(\R)$ be an arbitrary but fixed mapping such that
$\ins = \orig(\ins)\theta_{\orig(\ins)}$ with injective variable renaming $\theta_{\orig(\ins)}$.

\begin{example}[\contof{ex:instances}]\label{ex:variants}
    Let $\R = \{ \rho_1,\rho_2,\rho_3 \}$. Then $\var(\R) = \{ x,y,z,v \}$ and, e.g.,
    $\check{\rho}^a_1: t(v,z,y), q(v,z) \ruleTo p(v,y)$ and $\check{\rho}^b_1: t(x,y,x), q(x,y) \ruleTo p(x,x)$
    both are variants of $\rho_1$.
    In our example instance $\ins_1$ of rule $\rho_1$, the variables in the first and the third position of predicate $t$ are collapsed.
    So we can set $\orig(\ins_1) = \check{\rho}^b_1$
    with the injection $\theta_{\orig(\ins_1)}: x\mapsto x_1, y\mapsto y_1$, which yields $\ins_1 = \orig(\ins_1)\theta_{\orig(\ins_1)}$.
\end{example}

Moreover, given a set $\mathcal{F}$ of formulae $\exists\vec{v}.\, \psi$ where $\psi$ is a conjunction of atoms,
let $\subformulae(\mathcal{F}) \coloneqq \{ \exists \tilde{\vec{v}}.\, \tilde{\psi} \mid\exists\vec{v}.\, \psi\in\mathcal{F}, \tilde{\psi}\subseteq\psi, \tilde{\psi}\neq\varnothing, \tilde{\vec{v}}=\vec{v}\cap\var(\tilde{\psi}) \}$.

\begin{example}
    Consider the set $\subformulae(\pieces(\variants(\R)))$ of ``partial pieces of rule variants.'' The rule $a(x) \ruleTo \exists v.\, r(x,v), b(v)$ has a single piece, and contributes three partial pieces: (1) $\exists v.\, r(x,v), b(v)$,
    (2) $\exists v.\, r(x,v)$, and (3) $\exists v.\, b(v)$. 
    Another variant of the same rule is $a(x) \ruleTo r(x,x), b(x)$,
    which has two partial pieces $r(x,x)$ and $b(x)$ (but not $r(x,x), b(x)$, which is no piece).
\end{example}

To understand how we label a chain $c$, consider the chain rule $\rho_c$ and some rule $\rho\in\R$. 
To check $\rho_c\posr\rho$, Definition~\ref{def:posr} requires $\I_a$ and $\I_b$ such that:
(i) the positive body of $\rho$ matches $\I_b$ but not $\I_a$ (i.e., $\rho_c$ contributed something relevant to $\I_b$);
(ii) the negative body of $\rho$ is not matched in $\I_b$; and
(iii) applying $\rho$ adds something over $\I_b$ (match unsatisfied).
To check these, we store three components: 
($\labAname$) pieces of $\rho_c$ that $\underline{\opfont{c}}$ontribute new facts when applying $\rho_c$ ($\leadsto$ (i));
($\labBname$) partial pieces $\underline{\opfont{s}}$atisfied after applying $\rho_c$ ($\leadsto$ (iii)); and
($\labCname$) variants of rules inhibited by $\underline{\opfont{n}}$egative body atoms after applying $\rho_c$ ($\leadsto$ (ii)).

Now we can define the finite alphabet
$\Alphabet \coloneqq 2^{\mathcal{H}} \times 2^{\mathcal{F}} \times 2^{\variants(\R)}$
where $\mathcal{H} \coloneqq \pieces(\variants(\R))$ and $\mathcal{F} \coloneqq \subformulae(\mathcal{H})$.
To define the \emph{labelling function} $\ell$, consider a chain $c = \ins_1\ldots\ins_n$ with chain rule $\rho_c$,
and let $\theta_n\coloneqq \theta_{\orig(\ins_n)}$, i.e., $\ins_n=\orig(\ins_n)\theta_n$.
Moreover, let $\I_{c} \coloneqq \left(\bodyP(\rho_c)\cup\head(\rho_c)\right)\theta_n^{-1}\omega$,
which represents the necessary facts in $\I_b$ when checking $\rho_c\posr\rho$ for any $\rho\in\R$ (Definition~\ref{def:posr}),
in the notation of Definition~\ref{def:direct-posr} with the fixed mapping $\omega$.
We define $\ell(c) = \langle \labA(c), \labB(c), \labC(c) \rangle$ where:
\begin{align}
\labA(c) &\coloneqq \{ \exists \vec{v}.\tilde{\psi} \in \pieces(\orig(\ins_n)) \mid \tilde{\psi}\theta_n \not\subseteq \bodyP(\rho_c) \} \label{eq:alpha-label}\\
\labB(c) &\coloneqq \{ \exists\tilde{\vec{v}}.\tilde{\psi} \in \mathcal{F} \mid \I_{c} \models \exists\tilde{\vec{v}}.\tilde{\psi}\omega_\forall \} \label{eq:beta-label}\\
\labC(c) &\coloneqq \{ \check{\rho} \in \variants(\R) \mid \left(\I_{c}\cup\bodyP(\check{\rho})\omega\right) \cap \bodyN(\check{\rho})\omega \neq \varnothing \} \label{eq:gamma-label}
\end{align}
where $\omega_\forall$ is the restriction of $\omega$ to universal variables.
Note that $\labA(c)$ contains all head pieces of $\rho_c$ except possibly Datalog atoms.
The main correctness property of $\ell$ is this: if $\ell(c_1)=\ell(c_2)$ for chains $c_1$ and $c_2$, then for all $\rho\in\R$
and all ${\prec}\in\{\posr,\negr,\restr\}$, we have that $\rho_{c_1}\prec\rho$ iff $\rho_{c_2}\prec\rho$.
For $\posr$, the proof of this claim establishes that the three label components can indeed be used as in the intuition given above.
Cases for $\negr$ and $\restr$ are similar.
We can then show:

\newcounter{counterLemEquiLabelChains}\renewcommand{\thecounterLemEquiLabelChains}{(\roman{counterLemEquiLabelChains})}\begin{restatable}{lemma}{lemEquiLabelChains}\label{lem:equi-label-chains}
    Let $c_1,c_2\in\chains(\R)$ with $\ell(c_1)=\ell(c_2)$. If $c_1\ins_1\in\chains(\R)$ then there is $\ins_2\in\instances(\R)$ (both $\ins_1$ and $\ins_2$ belong to the same $\rho\in\R$), such that
    \begin{center}\begin{tabular*}{.8\textwidth}{@{\extracolsep{\fill}} l l l}
        {\normalfont\refstepcounter{counterLemEquiLabelChains}\label{lem:equi-label-chains:1}\thecounterLemEquiLabelChains} $c_2\ins_2\in\chains(\R)$  
        & and &
        {\normalfont\refstepcounter{counterLemEquiLabelChains}\label{lem:equi-label-chains:2}\thecounterLemEquiLabelChains} $\ell(c_1\ins_1)=\ell(c_2\ins_2)$ .
    \end{tabular*}\end{center}
    \setcounter{counterLemEquiLabelChains}{0} \end{restatable}

Part \ref{lem:equi-label-chains:1} largely follows from the correctness for $\posr$ by Definitions~\ref{def:extend} and \ref{def:chain}.
Part \ref{lem:equi-label-chains:2} is easy to see for $\labA$ and $\labC$; for $\labB$, it 
holds because $\mathcal{F}$ is closed under $\subformulae(\variants(\cdot))$ and
the instances added to the end of the chains are decoupled.

The \emph{word function} is $w: \chains(\R)\to\Alphabet^*$ with $c \mapsto \ell(c)$ for $|c|=1$,
$c \mapsto w(c_\pre)\ell(c)$ for $c=c_\pre\ins_n$.
The \emph{language of chains} is $\L(\R)=\{ w(c) \mid c\in\chains(\R) \}$.
Now given two words $u,v\in\L(\R)$ with $u_{|u|} = v_{|v|}$, $uw\in\L(\R)$ implies $vw\in\L(\R)$ for any $w\in\Alphabet^*$,
which follows from an inductive application of Lemma~\ref{lem:equi-label-chains}.
An equivalence relation on words based purely on their last letter therefore
fully characterises their \emph{extensions} into longer words. By Myhill-Nerode, we find that:

\begin{restatable}{theorem}{thmLrRegular}\label{thm:LR-regular}
    The language $\L(\R)$ is regular.
\end{restatable}

\def\Greliance{\mathcal{G}_\opfont{rel}}\def\assign{\mathrel{\mkern-.5mu\coloneqq\mkern-.5mu}}\SetKw{Continue}{Continue}
\SetKw{And}{and}
\SetKw{With}{with}
\SetKwFor{For}{for}{:}{}
\SetKwIF{If}{ElseIf}{Else}{if}{then :}{else if}{else}{}
\begin{algorithm}[t]
\caption{Is ruleset $\R$ chain-stratified?}\label{alg:chains}
$\Greliance \assign {\langle\R,\posr,\negr,\restr \rangle}$\tcp*[r]{compute the reliance graph}
    \lIf{$\Greliance$ has no cycle via ${\negr} \cup {\restr}$}{
        \Return ``$\R$ is fully stratified''
    }
    ${\bm{\prec}} \assign {{\negr} \cup {\restr}}$\tcp*[r]{precedence from $\Greliance$ (will approach ${\negrC}\cup{\restrC}$)}
    $\mathcal{C} \assign \{ \langle \rho,\rho,\rho,\epsilon \rangle \mid \rho \in \R$ and there is $\rho'\in\R$ with $\rho'\negr\rho$ or $\rho'\restr\rho$ in $\Greliance \}$;\label{line_cinit}\\\For{$\langle \rho_1,\rho_n, \rho_c, w(c_\pre) \rangle \in \mathcal{C}$ \And $\rho\in\R$ \With $\rho_n {\posr}\rho$ in $\Greliance$ \And $\rho_c {\posr}\rho$ \label{line_cloop}}{
        \lIf{$\bm{\prec}$ is cyclic}{\Return ``$\R$ is \underline{not} chain stratified''}\label{line_return_false}
        $\ins \assign $ injectively rename $\rho$, such that it does not use $\var(\rho_c)$;\\
        $V \assign \var(\head(\rho_c) \cup \bodyP(\ins))$;$\quad$ $C \assign \symbols(\head(\rho_c) \cup \bodyP(\ins))$;\\
        \For{$\eta: V\to V\cup C$ \With $\rho_c\eta \iposr \ins\eta$\label{line_etaloop}}{
            \lIf{$\ell(c\eta) \in w(c_\pre)$}{\Continue}\label{line_norepeat}
            $\rho_{c\eta \ins\eta} \assign (\bodyP(\rho_c)\cup\head(\rho_c)\cup\bodyP(\ins))\eta \ruleTo \exists \varE(\ins)\eta .\, \head(\ins)\eta$;\label{line_chainrule}\\
            $\mathcal{C} \assign \mathcal{C} \cup \{ \langle \rho_1, \rho, \rho_{c\eta \ins\eta}, w(c_\pre)\ell(c\eta) \rangle \}$;\label{line_addchain}\\
            \For{$\rho'\in\R$ \With $\rho \restr \rho'$ in $\Greliance$ \label{line_loop_restr}}{
                \lIf{$\rho_{c\eta \ins\eta} \restr \rho'$}{${\bm{\prec}} \assign {\bm{\prec}} \cup \{ \langle \rho_1, \rho' \rangle \}$ \label{line_new_restr}}
            }
            \For{$\rho'\in\R$ \With $\rho \negr \rho'$ in $\Greliance$ \label{line_loop_negr}}{
                \lIf{$\rho_{c\eta \ins\eta} \negr \rho'$}{${\bm{\prec}} \assign {\bm{\prec}} \cup \{ \langle \rho_1, \rho' \rangle \}$  \label{line_new_negr}}
            }
        }
    }
    \Return ``$\R$ is chain stratified'';\label{line_return_true}
\end{algorithm}
Algorithm~\ref{alg:chains} outlines a procedure that uses these insights to check chain stratification.
We represent chains $c=\ins_1\ldots\ins_n$ as 4-tuples $\tuple{\rho_1,\rho_n,\rho_c,w(\ins_1\ldots\ins_{n-1})}$.
We start with single-rule chains (line~A\ref{line_cinit}) for rules in the range of ${\negr}$ or ${\restr}$; this 
suffices since we search for ${\negrC\cup\restrC}$-cycles and ${\negrC\cup\restrC}$ only ranges over such rules.
We then iteratively extend $\mathcal{C}$ and $\bm{\prec}$ (A\ref{line_cloop}) until a cycle is found (A\ref{line_return_false})
or all chains have been considered and no cycle found (A\ref{line_return_true}).
For each $\rho_c\posr\rho$, we consider all concrete mappings $\eta$ such that $\rho_c\eta \iposr \ins\eta$ (A\ref{line_etaloop}),
if the $\eta$-specific chain's label was not encountered yet (A\ref{line_norepeat}).
Since chains are extended step-wise, presence of a label in $w(c_\pre)$ means a label-equivalent chain was already considered 
(and all related restraints and negative reliances found).
Then we construct the extended chain rule (A\ref{line_chainrule}), add the new chain (A\ref{line_addchain}, this adds to the cases iterated in A\ref{line_cloop}),
and add new pairs to $\bm{\prec}$ (A\ref{line_loop_restr}--\ref{line_new_negr}).

Keeping a hash set of seen chain labels
and checking if labels of new chains are already present
alleviates the need to store $w(c)$ as a whole.
Chain labels can even be canonised, as Lemma~\ref{lem:equi-label-chains} generalises to chains with isomorphic labels.

\section{Further Improving Chain Stratification}\label{sec_refinements}

Chain stratification generalises full stratification, which controls the interaction of
value invention and negation, and solves, e.g., Problem~2 (\alglineref{rule_father}--\alglineref{rule_different_fathers})
of Section~\ref{sec:intro}. Some generalisations, outlined next, allow us to cover more cases, e.g., Problem~1.

\myparagraph{Closure under constraints}
All reliances require a minimal set of facts (up to homomorphism), e.g., $\I_a$ in Def.~\ref{def:posr}. 
If the Datalog rules $\RpD\subseteq\R$ entail $\bot$ on these facts, then the reliance can be discarded \cite{MKH13:reliances}.
For chains $c$, we could apply $\RpD$ to $\I_c$ (see Section~\ref{sec:reg-lang}),
but for Algorithm~\ref{alg:chains} to be correct, we also must ensure that the label (not the specific chain) determines if
$\bot$ is derived or not.

We therefore consider additional formulae $\mathcal{B}$ to capture partial matches of Datalog rule bodies:
$\mathcal{B} \coloneqq \subformulae(\{ \exists \vec{v} .\, \bodyP(\check{\rho}) \mid \check{\rho} \in \variants(\R)_D, \vec{v} \subseteq \varA(\check{\rho}) \})$.
In $\ell(c)$, we will replace $\labB(c)\subseteq\mathcal{F}$ \eqref{eq:beta-label} by $\labB_D(c)\subseteq\mathcal{F}\cup\mathcal{B}$.
We convert sets $S\subseteq\mathcal{F}\cup\mathcal{B}$ 
to databases $\opfont{toDb}(S)\coloneqq \bigcup
    \big\{ \tilde{\psi}\omega_\forall \mu_\exists \mid \exists\tilde{\vec{v}}.\tilde{\psi} \in S,
       \mu_\exists: v \mapsto \fresh(v, \exists\tilde{\vec{v}}.\tilde{\psi}) \big\}$,
using the mapping $\omega$ and an auxiliary injection $\fresh: \var(S) \times S \injTo \N \setminus (\var(\R)\omega)$.

For chain $c=\ins_1\ldots\ins_n$ with $\theta_i\coloneqq\theta_{\orig(\ins_i)}$,
we recursively define $\labB_D(c)$, where
we use
$\opfont{Pre}_c \coloneqq \varnothing$ if $n=1$ and $\opfont{Pre}_c \coloneqq \opfont{toDb}(\labB_D(\ins_1\ldots\ins_{n-1}))\omega^{-1}\theta_{n-1}$ otherwise:
\begin{align}
\begin{split}
    \labB_D(c) \coloneqq 
    \{ &
        \exists\tilde{\vec{v}}.\tilde{\psi} \in \mathcal{F} \cup \mathcal{B}
        \mid{} \\
      &  \chase((\opfont{Pre}_c \cup \bodyP(\ins_n) \cup \head(\ins_n)) \theta_n^{-1}\omega,\RpD)
        \models
        \exists\tilde{\vec{v}}.\tilde{\psi}\omega_\forall
    \}
\end{split}
\end{align}
Compared to \eqref{eq:beta-label}, we replaced $\I_c$ by a chase result based on the previous chain label and the new final rule,
which has a homomorphism to $\chase(\I_c,\RpD)$.
Relying on the previous label lets us lift Lemma~\ref{lem:equi-label-chains} and related correctness results to $\labB_D$.
We also recast the extended label as a rule:
\begin{equation}\label{eq:chain-rule-dlclosure}
    \rho_c^D : \opfont{toDb}(\labB_D(c))\omega^{-1}\theta_n \setminus \head(\ins_n) \ruleTo \exists \varE(\ins_n) .\, \head(\ins_n)
\end{equation}
An $\exists$-disjoint sequence $c$ can be \emph{extended} by $\ins$ \emph{under constraints} if $\rho^{D}_c \iposr \ins$.
It is a \emph{chain} if $\bot\notin\bodyP(\rho^{D}_c)$ and Definition~\ref{def:chain} holds.
We write $\rho_1 \restrCD \rho_2$ (respectively $\rho_1 \negrCD \rho_2$)
if there is $c$ with $\rho^{D}_c \restr \rho_2$ (respectively $\rho^{D}_c \negr \rho_2$).

\begin{definition}\label{def_chainstratRD}
    $\R$ is \emph{chain-stratified under constraints} if ${\negrCD} \cup {\restrCD}$ is acyclic.
\end{definition}

Definition~\ref{def_chainstratRD} generalises chain stratification since only fewer chains are considered.
Correctness is retained if the chase prioritises Datalog rules:

\begin{corollary}\label{cor_chainstratRD}
    If $\R$ is chain-stratified under constraints,
    a chase that respects ${\prec_D}\coloneqq {\negrCD} \cup {\restrCD} \cup (\RpD \times (\R\setminus \RpD))$
    is generating and alternative match-free.
\end{corollary}

The relation $\prec_D$
is acyclic if ${\negrCD} \cup {\restrCD}$ is,
since $\prec_D$ ranges over $\R\setminus \RpD$.

\newcommand{\alglinerule}[1]{\rho_{\ref{#1}}}
\newcommand{\alglineins}[1]{\ins_{\ref{#1}}}
\begin{example}\label{ex_problem1_solved}
We return to Problem~1 (\alglineref{rule_ex_intro_student}--\alglineref{rule_constraint_teacher_type}),
denoting the respective rules $\alglinerule{rule_ex_intro_student}, \ldots, \alglinerule{rule_constraint_teacher_type}$.
For $\R=\{\alglinerule{rule_ex_intro_student},\alglinerule{rule_ex_intro_negation},\alglinerule{rule_ex_intro_rdfs5},\alglinerule{rule_ex_intro_rdfs7}\}$,
we find reliances 
$\alglinerule{rule_ex_intro_student} \negr \alglinerule{rule_ex_intro_negation}$,
$\alglinerule{rule_ex_intro_rdfs5} \posr \alglinerule{rule_ex_intro_rdfs5}$,
$\alglinerule{rule_ex_intro_rdfs7} \negr \alglinerule{rule_ex_intro_negation}$, and
$\rho \posr \alglinerule{rule_ex_intro_rdfs7} \posr \rho$ for all $\rho\in\R$ (due to the all-variable triple \nthreeinline{?x ?p ?y} in \alglineref{rule_ex_intro_rdfs7}).
$\R$ is not fully stratified, e.g., due to $\alglinerule{rule_ex_intro_negation} \posr \alglinerule{rule_ex_intro_rdfs7} \posr \alglinerule{rule_ex_intro_student} \negr \alglinerule{rule_ex_intro_negation}$.

$\R$ is not chain stratified either.
For clarity, we substitute constants from the image of $\omega$ directly.
To construct a chain for $\alglinerule{rule_ex_intro_negation} \alglinerule{rule_ex_intro_rdfs7} \alglinerule{rule_ex_intro_student}$, consider instances
$\alglineins{rule_ex_intro_negation} = \alglinerule{rule_ex_intro_negation}$,
$\alglineins{rule_ex_intro_rdfs7}=\alglinerule{rule_ex_intro_rdfs7}[\text{\nthreeinline{?p/:examiner,?q/:participant,?x/?c,?y/?p}}]$, 
and $\alglineins{rule_ex_intro_student} = \alglinerule{rule_ex_intro_student}$.
Then $\alglineins{rule_ex_intro_negation}$ is a length-1 chain, and
$\alglineins{rule_ex_intro_negation}\alglineins{rule_ex_intro_rdfs7}$ is a chain
since $\alglineins{rule_ex_intro_negation} \iposr \alglineins{rule_ex_intro_rdfs7}$, 
with chain rule $\rho_{\ref{rule_ex_intro_negation},\ref{rule_ex_intro_rdfs7}}$:
\begin{nthreelisting}
{?c :teacher ?p . :examiner rdfs:subPropertyOf :participant . ?c :examiner ?p} |\mintlabel{rule_ex_intro_chain_rule}|
    => {?c :participant ?p}.
\end{nthreelisting}
We have $\rho_{\ref{rule_ex_intro_negation},\ref{rule_ex_intro_rdfs7}} \iposr \alglineins{rule_ex_intro_student}$,
so $\alglineins{rule_ex_intro_negation}\alglineins{rule_ex_intro_rdfs7}\alglineins{rule_ex_intro_student}$ is a chain,
hence $\alglinerule{rule_ex_intro_negation} \posrC \alglinerule{rule_ex_intro_student}$.
Similarly, $\rho_{\ref{rule_ex_intro_negation},\ref{rule_ex_intro_rdfs7},\ref{rule_ex_intro_student}} \negr \alglinerule{rule_ex_intro_negation}$,
so we have a cycle $\alglinerule{rule_ex_intro_negation} \negrC \alglinerule{rule_ex_intro_negation}$,
and $\R$ is not chain stratified.

However, the extended $\R=\{\alglinerule{rule_ex_intro_student}, \ldots, \alglinerule{rule_constraint_teacher_type}\}$ is chain stratified under constraints,
where $\RpD = \R \setminus \{\alglinerule{rule_ex_intro_negation}\}$.
The $\posr$-paths from $\alglinerule{rule_ex_intro_negation}$ to $\alglinerule{rule_ex_intro_student}$ that do not revisit $\alglinerule{rule_ex_intro_negation}$ or pass through $\alglinerule{rule_ex_intro_student}$ can be described by a regular expression
$\alglinerule{rule_ex_intro_negation} (\alglinerule{rule_ex_intro_rdfs7} \alglinerule{rule_ex_intro_rdfs5}^\ast)^\ast \alglinerule{rule_ex_intro_rdfs7} \alglinerule{rule_ex_intro_student}$,
and each such path corresponds to one or more chains.
Similarly, $\alglinerule{rule_ex_intro_negation}(\alglinerule{rule_ex_intro_rdfs7} \alglinerule{rule_ex_intro_rdfs5}^\ast)^\ast \alglinerule{rule_ex_intro_rdfs7}$, for $\alglinerule{rule_ex_intro_negation}$ to $\alglinerule{rule_ex_intro_rdfs7}$.
All chains from $\alglinerule{rule_ex_intro_negation}$ to $\alglinerule{rule_ex_intro_student}$ violate constraints $\alglinerule{rule_constraint_teacher_participant}$ or $\alglinerule{rule_constraint_teacher_spo}$,
and the chains from $\alglinerule{rule_ex_intro_negation}$ to $\alglinerule{rule_ex_intro_rdfs7}$ that negatively rely on $\alglinerule{rule_ex_intro_negation}$ violate constraint $\alglinerule{rule_constraint_teacher_type}$.
The above chain $\alglineins{rule_ex_intro_negation}\alglineins{rule_ex_intro_rdfs7}\alglineins{rule_ex_intro_student}$ will be discarded due to $\alglinerule{rule_constraint_teacher_participant}$.
Since \nthreeinline{rdfs:subPropertyOf} is transitive by $\alglinerule{rule_ex_intro_rdfs5}$,
longer chains of the form $\alglinerule{rule_ex_intro_negation} \alglinerule{rule_ex_intro_rdfs7}^+ \alglinerule{rule_ex_intro_student}$
entail the same constraint-violating fact.
Any alternative chain from $\alglinerule{rule_ex_intro_negation}[\text{\nthreeinline{?p/rdf:type}}]$ to $\alglinerule{rule_ex_intro_student}$
via $\alglinerule{rule_ex_intro_rdfs7}[\text{\nthreeinline{?q/rdfs:subPropertyOf}}]$,
entails \nthreeinline{?c rdfs:subPropertyOf rdf:type}.
Then any triple \nthreeinline{?x ?c :Student} entails \nthreeinline{?x rdf:type :Student},
which could create a negative feedback cycle for $\alglinerule{rule_ex_intro_negation}$.
Such chains are discarded due to $\alglinerule{rule_constraint_teacher_spo}$.
Similarly, any chain starting like $\alglinerule{rule_ex_intro_negation} \alglinerule{rule_ex_intro_rdfs7}^+ \alglinerule{rule_ex_intro_rdfs5} \ldots$ 
violates $\alglinerule{rule_constraint_teacher_spo}$.
All surviving chains of the form $\alglinerule{rule_ex_intro_negation} \alglinerule{rule_ex_intro_rdfs7}^\ast \alglinerule{rule_ex_intro_rdfs7}[\text{\nthreeinline{?q/rdf:type}}]$,
whose chain rules negatively rely on $\alglinerule{rule_ex_intro_negation}$,
are eliminated by $\alglinerule{rule_constraint_teacher_type}$.
The remaining chains (all of the form $\alglinerule{rule_ex_intro_negation} \alglinerule{rule_ex_intro_rdfs7}^+$) are unproblematic.

Considering chains is still crucial here.
For a ``full stratification under constraints'' (following \cite{MKH13:reliances}),
we would need to prevent $\alglinerule{rule_ex_intro_negation} \posr \alglinerule{rule_ex_intro_rdfs7}$
with a stronger constraint such as 
\nthreeinline{{:examiner rdfs:subPropertyOf ?q} => false},
which would forbid useful triples, e.g., \nthreeinline{:examiner rdfs:subPropertyOf :teacher}.
\end{example}

\myparagraph{Null awareness}
When building chain rules as in \eqref{eq:chain-rule} (or \eqref{eq:chain-rule-dlclosure} under constraint closure),
the information which variables are existential is lost for all but the last instance.
To prohibit them from unifying with constants,
we can restrict a fixed arbitrary injection $\V\injTo\N$ to $\varE(\ins_{k-1})$ and apply it to the chain rule.
The latter may thus contain nulls, which can eliminate some irrelevant reliances.

\myparagraph{Negation awareness}
Keeping negations of the atoms $\bigcup_{i\in\set[k]} \bodyN(\ins_i)$
in the chain rule's body sharpens the analysis further.
They may prevent some outgoing reliances of the chain rule.
However, this needs to be reflected in $\labC(c)$,
tracking the subset of these atoms that only use frontier variables.
This is clearly bounded.
It suffices, as only such atoms may unify with facts of the representative databases
examined for reliance checks between the chain rule and other rules.

\ifTechnicalReport
\section{Implementation \& Evaluation}\label{sec:impl}

We implement null- and negation-aware chain stratification under $\RpD$-closure
in a Rust library\footnote{
Source code (ca. 13k LoC) available
in the supplementary material statement.
}
and a PoC tool using it,
which accepts Nemo \cite{Ivliev+:Nemo2024} and VLog \cite{UJK:VLog2016} syntax,
as well as (a subset of) N3 syntax (via a Python transpiler).

While our method is designed to address known problematic cases that arise in RDF rules,
there are no representative RDF-based rulesets that could be used to evaluate its \emph{effectiveness}.
However, we can investigate the feasibility of our analysis in terms of \emph{performance}.
To this end, we use our tool to check chain stratification on a benchmark that was also
used by González~\textit{et al}. \cite{GIKM2022} in their analysis of core stratification,
which consists of 201 piece-decomposed, negation-free rulesets with predicate symbols (not just triple).\footnote{They correspond to a subset of the \emph{\href{https://www.cs.ox.ac.uk/isg/ontologies/}{Oxford Ontology Repository}}, accessed 2025-12-01.}
Without negation, core and full chain stratification coincide, but the effort of
computing chains is still realistic.
We ran our experiments on a Linux server (2$\times$QuadCore Intel Xeon 3.5GHz, 768GiB RAM),
with a 15min timeout per analysis.
Each run is repeated thrice and the median time values are reported in our supplementary material.

Overall, $80$ rulesets were classified in less than 0.1sec, $115$ in less than 1sec, $160$ in under 1 min,
and $187$ within the total timeout. The remaining $14$ rulesets that did not finish
in that time each contained over 60,000 rules.
$128$ cases were not core-stratified (the previous best condition for negation-free rules), and required
the analysis of chains. These cases on average required an additional $12.51\%$ of analysis time
in comparison to the check for core stratification.
Considering the relative complexity of our definitions, we consider this to be a very acceptable overhead
that appears to be feasible in practice, especially since many sets of RDF rules are not covered
by any other conditions.

 \fi

\section{Conclusion}\label{sec:conclude}

We established chain stratification to improve beyond prior notions. Augmented by constraints, our approach enables the use of N3 rules like \emph{rdfs7} in combination with blank nodes and negation.
We designed a refined criterion based on the transitive closure of positive reliances, contributed
interesting theoretical results on their decidability, and provided a prototypical implementation in Rust.
Promising directions of future work may include
(1) integration with RDF rule engines, (2) repair of non-chain-stratified rulesets by mining constraints from counter examples generated by reliance computations,
(3) transfer of our finding to the handling of aggregates, and
(4) analysis of the utility of $\posrC$ to establish new termination criteria. Finally, we hope that our insights can also contribute to ongoing standardisation activities of RDF rules.

\begin{credits}

\subsubsection{\ackname}
This work is supported by
Deutsche Forschungsgemeinschaft (DFG, German Research Foundation)
in project number 389792660 (TRR 248, \href{https://www.perspicuous-computing.science/}{Center for Perspicuous Systems})
and 390696704 (CeTI Cluster of Excellence),
by the Bundesministerium für Forschung, Technologie und Raumfahrt (BMFTR, Federal Ministry of Research, Technology and Space)
in the \href{https://www.scads.de/}{Center for Scalable Data Analytics and Artificial Intelligence} (ScaDS.AI), and
in DAAD project 57616814 (\href{https://secai.org/}{SECAI, School of Embedded Composite AI})
as part of the program Konrad Zuse Schools of Excellence in Artificial Intelligence.

\subsubsection{Supplemental Material Statement.}
The source code of our Rust implementation is published at \url{\gitlabLink}.
\ifTechnicalReport
Refer to \texttt{README.md} to reproduce the results.
\fi

\subsubsection{Declaration of use of Generative AI.}
During the preparation of this work the authors used ChatGPT to ask for Rust programming tips.
Opus 4.6 was used for proofreading.
The text in this paper as well as the source code is fully handwritten.

\end{credits}
 
\bibliographystyle{splncs04}

\ifTechnicalReport
\clearpage

\appendix

\section{Proofs for Section~\ref{sec:trails}}

In the ensuing proofs, we will use the following additional notation:
\begin{itemize}
    \item $\matches(\I)$ for the set of all matches of rules from $\R$ over database $\I$.
    \item $\unsat(\R) \subseteq \matches(\I)$ to collect all the unsatisfied matches over $\I$. \item $\cHom(C,k)$ for $\mu_\forall \mu_\exists$ in the extended match $\cMatch(C,k)=\langle \cRule(C,k), \mu_\forall \mu_\exists \rangle$ applied in step $k$ of chase $C$.
\end{itemize}

\thmPosrTPrecNoViolate*

We will show a slightly stronger version of the theorem for $\posrT\circ\prec$, assuming only that $\negr\subseteq\prec$.
(The specifics of $\restr$ are irrelevant for the proof.)

\begin{proof}[of contrapositive of Theorem~\ref{thm:posrT-prec-no-violate}]
    Let $C$ be a chase sequence that violates $\prec$.
    Then there must be chase steps $i\leq k$ with $\cRule(C,k)\prec \cRule(C,i)$.
    In fact, $i<k$ as there would otherwise be a self-loop in $\posrC\circ\prec$.
    Select such $\langle i,k\rangle$ (lexicographic) minimally.

    Define the following sequence of indices:
    \begin{align}
        a_1 &\coloneqq k \\
        a_{x+1} &\coloneqq \min \{ j \mid \cMatch(C,a_x) \in \matches(\cDB(C,j)) \} \qquad \text{if } i < a_x
    \end{align}
    This sequence is clearly strictly decreasing (and finite). Also, it has length $|a| > 1$.
    For all $x<|a|$, we have $i < a_x$. Only $a_{|a|}$ may be smaller or equal to $i$.

    Due to the minimality of $\langle i,k\rangle$, and because of ${\negr} \subseteq {\prec}$, no negative reliance can be violated in $C$ up to step $k$.
    Therefore, $C$ is generating up to $k$.
    Hence, we know that a match selected in step $a_x$
    that was already available in step $j<a_x$ must also be available in all intermediate steps.

    The second-to-last match $\cMatch(C,a_{|a|-1})$ was already applicable in step $a_{|a|} \leq i$.
    We therefore clearly have $\cMatch(C,a_{|a|-1}) \in \matches(\cDB(C,i))$.
    \begin{center}

\begin{tikzpicture}[>=stealth]

\draw[->] (-1.3,0) -- (5,0);
\node[below] at (-0.5,0) {chase seq.};

\draw (1,-0.1) -- (1,0.1);
\draw (4,-0.1) -- (4,0.1);

\node[below] at (1,0) {$i$};
\node[below] at (4,0) {$k$};

\draw[->, bend right=40] (4,0.2) to node[above,align=center] {violates\\$\cRule(C,k)\prec\cRule(C,i)$} (1,0.2);

\node[below] at (4,-0.6) {$a_1$};
\node[below] at (3.2,-0.6) {$a_2$};
\draw[->, bend left=20] (4,-0.5) to (3.2,-0.5);

\node[below] at (2.5,-0.6) {$a_3$};
\draw[->, bend left=20] (3.2,-0.5) to (2.5,-0.5);

\node[below] at (1.9,-0.4) {$\ldots$};
\node[below] at (0.6,-0.6) {$a_{|a|}$};
\draw[->, bend left=25] (1.3,-0.5) to (0.6,-0.5);

\end{tikzpicture}

     \end{center}
    The sequence of rule instances $t=\ins_1,\ldots,\ins_{|a|}$ with
    \begin{equation}
        \ins_x = \cRule(C,a_{|a|+1-x})\cHom(C,a_{|a|+1-x})\omega^{-1}\qquad\text{for }x\in\set[|a|]
    \end{equation}
    is a trail,
    because the chase $C$ with step mapping $f: \set[|a|] \injTo \cSteps(C)$, $x \xmapsto{f} a_{|a|+1-x}$
    fulfils the two conditions from Definition~\ref{def:trail}.
    Match correspondence follows from the definition of $\ins_x$ based on the matches applied in step $f(x)$ of $C$.
    Causal connection is a consequence of the subsequent sequence item $a_{x+1}$ being the minimal chase step where the match used at $a_x$ was available.
    This means that $\cMatch(C,a_x)$ must map a body atom of $\cRule(C,a_x)$ onto a fact added by applying $\cMatch(C,a_{x+1})$, i.e. $\head(\ins_x)\cap\bodyP(\ins_{x+1})\neq\varnothing$.

    We thus have $\cRule(C,a_{|a|}) \posrT \cRule(C,a_1)$.

    As we set $a_1=k$, we have $\cRule(C,a_1)=\cRule(C,k)$ and we hence $\cRule(C,a_1)\prec \cRule(C,i)$.
    This gives $\cRule(C,a_{|a|}) \posrT \circ \prec \cRule(C,i)$ and both rules were selectable in chase step $i$.
    We also know that $\cMatch(C,a_{|a|}) \neq \cMatch(C,i)$, because $\posrT\circ \prec$ would otherwise have a self-loop.
    Hence, $C$ disrespects $\posrT \circ \prec$,
    as $\cMatch(C,a_{|a|})$ should have been preferred over $\cMatch(C,i)$ in chase step $i$.
\qed
\end{proof}

\thmPosrTPrecUniqeChaseResult*

\begin{proof}[of contrapositive of Theorem~\ref{thm:posrT-prec-unique-chase-result}]
    Let $C_1$ and $C_2$ be two generating chases of (the same) $\KB$ with non-isomorphic results.
    This can only be the case if there is (w.l.o.g) a match that is applied in $C_1$ but never applied in $C_2$.
    Further, there must be such a match, whose application in $C_1$ has added facts which are not present in $C_2$ ($\ddagger$).
    Select a minimal $i\in\cSteps(C_1)$, such that $\cMatch(C_1,i)$ is such a match.
    All matches applied before $i$ in $C_1$ must therefore be applied at some point in $C_2$.
    Let $i_2\in\cSteps(C_2)$ be the chase step at which the last of them was applied.
    The database $\cDB(C_1,i-1) = \I_0 \cup \bigcup_{k<i} \head(\cRule(C_1,k))\cHom(C_1,k)$
    must homomorphically map into $\cDB(C_2,i_2)$ ($\dagger$).
    Thus, any unsatisfied match over $\cDB(C_2,i_2)$ must either (1) also be unsatisfied over $\cDB(C_1,i-1)$ or (2) not be a match there (because the necessary facts were not derived yet).
    In case (2), it must be derivable by a sequence of applications starting with an unsatisfied match over $\cDB(C_1,i-1)$.
    The $\cMatch(C_1,i)$ may either be (a) unsatisfied, (b) satisfied, or (c) invalidated\footnote{by deriving a fact from $\bodyN(\cRule(C_1,i))\cHom(C_1,i)$} over $\cDB(C_2,i_2)$.
    If (a), then, there must be a $k\geq i_2$, such that $\cMatch(C_1,i)$ is no unsatisfied match over $\cDB(C_2,k)$ (due to fairness). Applying $\cMatch(C_2,k)$ has either (b) satisfied $\cMatch(C_1,i)$, or (c) invalidated it.
    In case (b), $\cRule(C_1,i)$ cannot be Datalog, as satisfying it would otherwise not have added facts in $C_1$ which $C_2$ lacks (as is required by $\ddagger$).
    Either $\cRule(C_1,i)$ has both Datalog and existential pieces, which makes it self-restraining.
    This would mean that $C_1$ does not respect the precedence ${\restr} \subseteq {{\posrT}\circ(\negr\cup\restr)}$.
    Or $\cRule(C_1,i)$ has only existential pieces, which means that there is an alternative match for $\cMatch(C_1,i)$ over $\cDB(C_2,k)$,
    i.e. $\cRule(C_2,k)\restr\cRule(C_1,i)$ by Definition~\ref{def:restr}.
    In case (c), $\cRule(C_2,k)\negr\cRule(C_1,i)$ by Definition~\ref{def:negr}.
    Similar to the proof of Theorem~\ref{thm:posrT-prec-no-violate} above,
    we construct a sequence of chase steps from $C_2$ like
    \begin{align}
        a_1 &\coloneqq k \\
        a_{x+1} &\coloneqq \min \{ j \mid \cMatch(C_2,a_x) \in \matches(\cDB(C_2,j)) \}
    \end{align}
    It is finite, as it is clearly strictly decreasing and all sequence items are non-negative. In fact, it ends with $a_{|a|} = 0$.
    By ($\dagger$) and because $C_2$ is generating, we must be able to find an $x$, such that either (1) $\cMatch(C_2,a_x) \in \unsat(\cDB(C_1,i-1))$,
    or (2) there exists a $\langle\rho,\mu\rangle \in \unsat(\cDB(C_1,i-1))$ with appropriate step mapping for $C_2$
    witnessing that $t_\pre = \rho\mu\omega^{-1}\ldots \cRule(C_2,a_x)\cHom(C_2,a_x)\omega^{-1}$ is a trail.
    We then use $a_x \ldots a_1$ to define a sequence of rule instances $t$ and step mapping $f$.
    By Definition~\ref{def:trail}, $C_2$ and $f$ witness that $t$ is a trail.
    So we now have that $\cRule(C_2,a_x) \posrT \cRule(C_2,k) {(\negr \cup \restr)} \cRule(C_1,i)$.
    If (1), there is an unsatisfied match over $\cDB(C_1,i-1)$ for $\cRule(C_2,a_x)$ (and clearly also for $\cRule(C,i)$).
    As ${\posrT}\circ(\negr\cup\restr)$ is acyclic by assumption, $\cRule(C,i) \neq \cRule(C_2,a_x)$.
    So $C_1$ does not respect ${\posrT} \circ (\negr \cup \restr)$,
    as $\cRule(C_2,a_x)$ should have been preferred over $\cRule(C_1,i)$.
    If (2), $\rho \posrT \cRule(C_2,a_x) \posrT \cRule(C_2,k) {(\negr \cup \restr)} \cRule(C_1,i)$. 
    The sequence $t_\pre t$ is also a trail, i.e. $\rho \posrT \cRule(C_2,k) {(\negr \cup \restr)} \cRule(C_1,i)$ 
    and $\langle \rho,\mu\rangle$ should have been preferred over $\cRule(C_1,i)$. Again, $C_1$ disrespects ${\posrT} \circ (\negr \cup \restr)$.
\qed
\end{proof}

\lemPosrTSmallest*

\begin{proof}
    Let ${<} \subset {\posrT}$ be a smaller relation.
    Then there must be some $\rho_1\posrT\rho_2$ with $\rho_1\not<\rho_2$.
    Consider a relation ${\prec} \supseteq {\negr}$, such that there is some $\rho_3$,
    such that $\rho_2\prec\rho_3$ and $< \circ \prec$ is acyclic.
    Suppose towards contradiction that any chase respecting $< \circ \prec$ has no $\prec$ violations.
    Consider a chase sequence $C$ that violates $\rho_2\prec\rho_3$, but no other $\prec$-edges.
    Analogous to the proof of Theorem~\ref{thm:posrT-prec-no-violate},
    select $i<k$ minimally, such that $\cRule(C,k)\prec \cRule(C,i)$
    (in fact, $\cRule(C,k) = \rho_2$ and $\cRule(C,i) = \rho_3$),
    then construct the sequence $(a_x)_{x=1,\ldots |a|}$, and the trail $t=\rho_1,\ldots,\rho_{|a|}$.
    In chase step $i$, both $\cMatch(C,i)$ and $\cMatch(C,a_{|a|})$ are selectable, but not $\cMatch(C,k)$.
    As $\rho_1\not<\rho_2$, we do not have $\cRule(C,a_|a|) = \rho_1 < \circ \prec \rho_3 = \cRule(C,i)$,
    so $C$ respects $<\circ\prec$.
\qed
\end{proof}

\thmTrailUndecidable*

\begin{proof}[by reduction from the halting problem]
    Let $M$ be a Turing machine with initial state $q_0$ and single ``accept'' state $q_f$ and let $w$ be a word.
    It is undecidable whether $M$ has a finite run on $w$ that reaches the ``accept'' state.
    \begin{itemize}
        \item Turing machine $M=\langle Q, \Sigma, \Gamma, \delta, q_0, q_f \rangle$ with
            set of states $Q$,
            input alphabet $\Sigma$,
            tape alphabet $\Gamma\supseteq \Sigma$,
            (partial) transition function $\delta: Q\times\Gamma \to Q\times\Gamma\times\{L,R\}$,
            initial state $q_0$ and
            final state $q_f$.
        \item Elements of $\Gamma^*\times Q\times \Gamma^*$ are configurations of $M$.
            For $\delta(q,a)=\langle q',a',d\rangle$,
            write $\langle xa,q,by\rangle \vdash \langle xa'b,q',y\rangle$ iff $d=R$
            and $\langle xa,q,by\rangle \vdash \langle x,q',a'by\rangle$ iff $d=L$.
        \item For word $w\in\Sigma^*$, a run of $M$ is $\langle w,q_0,\epsilon\rangle \vdash \ldots$.
            If this run is finite and ends with $q_f$, then $M$ accepts $w$.
    \end{itemize}
    Encode $M$ and $w$ into a ruleset $\R(M,w)$
    with two designated rules $\rho_1, \rho_2 \in \R(M,w)$, such that $\rho_1 \posrT \rho_2$ if and only if $M$ holds and accepts $w$.
    In detail:\\\hphantom{.}

    \noindent
    \textbf{Input:} $\langle M, w\rangle$ for word $w\in\Sigma^*$\\
    \textbf{Output:} $\R(M,w)$\hfill(constants in bold)
    \begin{align*}
        \rho_1 :\quad & \ruleTo \exists t_0,x_1,\ldots,x_{|w|} .\, \emph{state}(t_0,\bm{q_0}), \emph{pos}(t_0,x_1), \emph{tape}(t_0,x_1,\bm{w_1}), \emph{next}(x_1,x_2), \\
             &\qquad\quad \ldots, \emph{tape}(t_0,x_{|w|-1},\bm{w_{|w|-1}}), \emph{next}(x_{|w|-1},x_{|w|}), \emph{tape}(t_0,x_{|w|},\bm{w_{|w|}})\\
        & \emph{state}(t,s) \ruleTo \exists t'.\, \emph{step}(t,t')\\
        & \emph{tape}(t,x,c) \ruleTo \exists x'.\, \emph{next}(x,x')\\
        & \emph{tape}(t,x,c), \emph{pos}(t,x'), x\neq x', \emph{step}(t,t') \ruleTo \emph{tape}(t',x,c)\\
        & \emph{state}(t,\bm{s}), \emph{pos}(t,x), \emph{tape}(t,x,\bm{c}), \emph{next}(x',x), \emph{step}(t,t') \\
            &\qquad \ruleTo \emph{state}(t',\bm{s'}), \emph{tape}(t',x,\bm{c'}), \emph{pos}(t',x') \qquad\quad\text{for } \delta(\bm{s},\bm{c})=\langle \bm{s'},\bm{c'},L\rangle\\
        & \emph{state}(t,\bm{s}), \emph{pos}(t,x), \emph{tape}(t,x,\bm{c}), \emph{next}(x,x'), \emph{step}(t,t') \\
            &\qquad \ruleTo \emph{state}(t',\bm{s'}), \emph{tape}(t',x,\bm{c'}), \emph{pos}(t',x') \qquad\quad\text{for } \delta(\bm{s},\bm{c})=\langle \bm{s'},\bm{c'},R\rangle\\
        \rho_2 :\quad & \emph{state}(t_0,\bm{q_0}), \emph{state}(t,\bm{q_f}) \ruleTo \emph{accept}()
    \end{align*}

    A restricted chase $C$ of the $\I$ and $\R(M,w)$ is a faithful simulation, irrespective of $\I$.
    (Even if $\I \neq \varnothing$ and for instance contains a cyclic $\emph{tape}$ or $\emph{step}$ predicate, this is no problem,
    as the initialization rule invents fresh nulls without connection to atoms in $\I$.)
    In case $M$ holds and accepts $w$, the sequence $t=\rho_1\rho_2$ is a trail,
    witnessed by some such chase (representing the finite accepting run) and a corresponding step mapping.
    (The causal connection property is always met due to the $\emph{state}(t_0,\bm{q_0})$ atoms occurring both in $\head(\rho_1)$ and $\bodyP(\rho_2)$.)
    Otherwise, no chase with appropriate step mapping exists and $t$ cannot be a trail.
\qed
\end{proof}
 \section{Proof for Section~\ref{sec:chains}}

\lemTrailGenDecChain*

\begin{proof}
    Let $t=\ins_1\ldots\ins_n$ be a trail.
    By Definition~\ref{def:trail}, $t$ is $\exists$-disjoint and
    there must be a witnessing chase $C$ with step mapping $f: \set[n]\injTo\cSteps(C)$,
    such that $\cMatch(C,f(x)) = \langle\ins_x,\omega\rangle$ and $\head(\ins_x)\cap\bodyP(\ins_{x+1})\neq\varnothing$.

    For each $i\in\set[n]$, consider an injective replacement $\theta_1=\id$ and  for $i>1$
    $\theta_i: \varA(\ins_i)\setminus\var(\head(\ins_{i-1})) \injTo \V \setminus (\bigcup_{j\in\set[i]} \var(\ins_j))$
    of stale variables with fresh ones (as in Example~\ref{ex:generalised}).
    The sequence $\decoupled(t)=\ins_1\theta_1,\ldots,\ins_n\theta_n$ is clearly decoupled (and remains $\exists$-disjoint).
    It generalises $t$ with $\sigma = \theta_1^{-1}\circ\ldots\circ\theta_n^{-1}$.
    (The individual replacements are invertible because they are injective.
    This function composition assumes that the domains of $\theta_i^{-1}$ are extended to $\V$ by identity in the usual way.)

    To show that $\decoupled(t)$ is a chain, by Definition~\ref{def:chain},
    we have to show for each $k\in\set[n-1]$ that $\rho_{\ins_1\theta_1\ldots\ins_k\theta_k} \iposr \ins_{k+1}\theta_{k+1}$.
    We have $\ins_i\theta_i\sigma = \ins_i$ for all $i\in\set[n]$, and
    hence also $\rho_{\ins_1\theta_1\ldots\ins_i\theta_i}\sigma = \rho_{\ins_1\ldots\ins_n}$ for the corresponding chain rules (see Equation~\eqref{eq:chain-rule}).
    So $\langle \rho_i\theta_i, \omega\rangle \in \matches(\I\omega^{-1}\theta_i\omega)$ iff $\langle \rho_i, \omega\rangle \in \matches(\I)$ and
    $\langle \rho_{\ins_1\theta_1\ldots\ins_i\theta_i}, \omega\rangle \in \matches(\bigcup_{j\in\set[i]} \I\omega^{-1}\theta_j\omega)$
    iff $\langle \rho_{\ins_1\ldots\ins_n}, \omega\rangle \in \matches(\I)$.
    Since $\I\omega^{-1}\theta_i\omega \subseteq \bigcup_{j\in\set[i]} \I\omega^{-1}\theta_i\omega$ and
    none of the negative body atoms of the matching rules can match in any of the copies,
    we further have that $\matches(\I\omega^{-1}\theta_i\omega) \subseteq \matches(\bigcup_{j\in\set[i]} \I\omega^{-1}\theta_j\omega)$.

    Let $k\in\set[n-1]$ and define
    \begin{align*}
        \I_\delta &\coloneqq \left(\head(\rho_{\ins_1\theta_1\ldots\ins_k\theta_k}) \setminus \bodyP(\rho_{\ins_1\theta_1\ldots\ins_k\theta_k})\right)\omega \text{,}\\
        \I_b &\coloneqq \textstyle\bigcup_{j\in\set[k]} \cDB(C,f(k))\omega^{-1}\theta_j\omega \text{, and}\\
        \I_a &\coloneqq \I_b \setminus \I_\delta \text{.}
    \end{align*}
    Now consider the matches
    $\lambda_1 \coloneqq \langle \rho_{\ins_1\theta_1\ldots\ins_k\theta_k}, \omega \rangle$ and
    $\lambda_2 \coloneqq \langle \ins_{k+1}\theta_{k+1},\omega \rangle$.
We have $\lambda_1 \in \matches(\bigcup_{j\in\set[k]} \cDB(C,f(k))\omega^{-1}\theta_j\omega) = \matches(\I_b)$ (see above).
    Removal of $\I_\delta$ cannot have removed the match, so $\lambda_1\in\matches(\I_b\setminus \I_\delta)=\matches(\I_a)$.
    The match $\lambda_1$ is unsatisfied over $\I_a$, as $\I_\delta$ is not contained in $\I_a$.
    Thus, $\lambda_1 \in \unsat(\I_a)$ and
    the database $\I_b$ was obtained by applying $\lambda_1$ to $\I_a$.

    Analogously, we have $\lambda_2 \in \unsat(\cDB(C,f(k))\omega^{-1}\theta_k\omega)$, i.e. $\lambda_2 \in \unsat(\I_b)$.
    Also, $\lambda_2 \not\in \matches(\I_a)$, as it requires facts from $\I_\delta$ that are only present in $\I_b$.
    Hence, $\lambda_2 \in \unsat(\I_b) \setminus \matches(\I_a)$. \qed
\end{proof}
 \section{Proofs for Section~\ref{sec:chain-strat}}

\lemTrailRestrC*

\begin{proof}
    To show $\ins_1 \restrC \cRule(C,i)$, by Definition~\ref{def:restrC}
    we need to give a decoupled chain $c$ that starts with $\ins_1$, such that $\rho_c \restr \cRule(C,i)$.
    Obtain $c\in\chains(\R)$ generalising $t$ with $\sigma$, by injectively renaming stale variables as in Lemma~\ref{lem:trail-gen-dec-chain}.
    By Definition~\ref{def:restr}, we thus need to give databases $\I_a\subseteq \I_b$ and mappings $\mu_1, \mu_2$,
    with $\langle \cRule(C,i), \mu_2\rangle \in \unsat(\I'_a)$ and
    $\I_a = \I'_a \cup \head(\cRule(C,i))\mu_2$ and
    $\langle \rho_c, \mu_1 \rangle \in \unsat(\I_a)$ and
    $\I_b = \I'_b \cup \head(\rho_c)\mu_1$,
    such that there is alternative match for $\langle \cRule(C,i), \mu_2\rangle$ over $\I_b$ but not over $\I'_b$.
    Take $\I'_a = \cDB(C,i-1)$ and $\I_a = \cDB(C,i)$ with $\mu_2 = \cHom(C,i)$.
    We clearly have $\langle\cRule(C,i), \mu_2\rangle \in \unsat(\I'_a)$ and $\I_a = \I'_a \cup \head(\cRule(C,i))\mu_2$ by Definition~\ref{def:chase}.
    Take $\I'_b = \cDB(C,f(n)-1)$ and $\I_b = \cDB(C,f(n))$. It is clear that $\I_a \subseteq \I_b$, since $i \leq f(n)-1$.
    Recall that by Equation~\eqref{eq:chain-rule},
    $\bodyP(\rho_c) = \left(\bigcup_{i\in\set[n]} \bodyP(\ins_i)\right)\cup\left(\bigcup_{i\in\set[n-1]} \head(\ins_i)\right)$ and
    $\head(\rho_c) = \head(\ins_n)$.
    We now have:
    \begin{itemize}
        \item $\bodyP(\rho_c)\sigma\omega \subseteq \I'_b$ and $\I'_b \not\models \head(\rho_c) \sigma\omega_\forall$,
            i.e. $\langle \rho_c, \sigma\omega \rangle \in \unsat(\I'_b)$, and
        \item $\I_b = \cDB(C,f(n)+1) = \cDB(C,f(n)) \cup \head(\cRule(C,f(n)))\cHom(C,f(n))$\\
                \hphantom{$\I_b = \cDB(C,f(n)+1)$} $= \I'_b \cup \head(\ins_n)\omega = \I'_b \cup \head(\rho_c)\omega$.
    \end{itemize}
    By assumption, we have an alternative match for $\cRule(C,i)$ over $\I_b$ but not $\I'_b$.
\qed
\end{proof}

\lemTrailNegrC*

\begin{proof}
    Analogous to the above proof of Lemma~\ref{lem:trail-restrC},
    we show by Definition~\ref{def:negrC} $\ins_1 \negrC \cRule(C,i)$ via chain $c$ generalising $t$ with $\sigma$ as in Lemma~\ref{lem:trail-gen-dec-chain}.
    By Definition~\ref{def:negr} we need to give database $\I_a$ and mappings $\mu_1, \mu_2$
    with $\langle \rho_c,\mu_1 \rangle \in \unsat(\I_a)$ and $\I_b = \I_a \cup \head(\rho_c)\mu_1$,
    such that $\langle \cRule(C,i),\mu_2 \rangle \in \unsat(\I_a) \setminus \unsat(\I_b)$.
    This is the case for $\I_a = \cDB(C,f(n)-1)$, $\mu_1 = \sigma\omega$ and $\mu_2=\cHom(C,i)$.
\qed
\end{proof}
 \section{Proofs for Section~\ref{sec:reg-lang}}

\lemEquiLabelChains*

\noindent
We will use some auxiliary notation:
\begin{itemize}
    \item
    When we say formula set, we mean a set of formulae of the form $\exists \vec{v}.\psi$ for conjunction of atoms $\psi$.
    This is essentially a set of pieces.
    \item
    For databases $\I_1$ and $\I_2$ and formula set $\mathcal{F}$, we write $I_1 \equiv_\mathcal{F} I_2$,
    if for all formulae $\exists \vec{v}.\tilde{\psi} \in \mathcal{F}$, we have that
    $\I_1 \models \exists \vec{v}.\tilde{\psi}\omega_\forall$ iff $\I_2 \models \exists \vec{v}.\tilde{\psi}\omega_\forall$.
    \item
    For database $\I$, we write $\symbols(\I)$ for the set of constants and nulls occurring in facts of $\I$.\\
    For formula set $\mathcal{F}$, we write
    $\symbols(\mathcal{F}) = \bigcup_{\exists \vec{v}.\psi \in \mathcal{F}} \varA(\exists \vec{v}.\psi)\omega$
    for the set of constants and nulls injectively assigned to universal variables of $\mathcal{F}$ by $\omega$.
    \item
    We will apply $\variants(\cdot)$ to formula sets, defined exactly like $\variants(\R)$ (page \pageref{eq:variants}) but using $\mathcal{F}$ instead of $\R$.
\end{itemize}

\begin{proof}
    Let $c_1$ and $c_2$ be two decoupled chains with the same label.
    Then $\labA(c_1) = \labA(c_2)$, $\labB(c_1) = \labB(c_2)$, and $\labC(c_1) = \labC(c_2)$.
    Let $\ins_n$ and $\ins_m$ be the last rule instances of $c_1$ and $c_2$, respectively.
    Further, let $\theta_n$ and $\theta_m$ be the two (unique) injective mappings,
    such that $\ins_n = \orig(\ins_n)\theta_n$ and $\ins_m = \orig(\ins_m)\theta_m$.

    Let $\ins_1 \in \instances(\R)$, such that $c_1\ins_1$ is a decoupled chain.
    By Definition~\ref{def:decoupled}, $c_1\ins_1$ is decoupled if and only if all variables from $\var(\ins_1) \setminus \var(\head(\ins_n))$ are fresh w.r.t. $c_1$.
    By Definitions~\ref{def:chain}, $c_1\ins_1$ is a chain if and only if $\rho_{c_1} \iposr \ins_1$.
    Since $\theta_n$ is injective, we can apply its inverse on both sides and get $\rho_{c_1} \theta_n^{-1} \iposr \ins_1 \theta_n^{-1}$,
    which means by Definition~\ref{def:direct-posr} that,
    for the database
    $\I_1 \coloneqq (\bodyP(\rho_{c_1}) \cup \head(\rho_{c_1}) \cup \bodyP(\ins_1)) \theta_n^{-1}\omega$, we have
    \begin{enumerate}[leftmargin=.9cm,label=(1.\alph*)]
        \item\label{it:c1-alpha} the intersection between the atom sets $(\head(\rho_{c_1}) \setminus \bodyP(\rho_{c_1})) \theta_n^{-1}$ and $\bodyP(\ins_1) \theta_n^{-1}$ is non-empty, and
        \item\label{it:c1-beta} the database $\I_1$ does not satisfy $\exists \vec{v} .\, \head(\ins_1)\theta_n^{-1}\omega_\forall$, and
        \item\label{it:c1-gamma} none of the negative body atoms of $\ins_1\theta_n^{-1}$ are in $\I_1$.
    \end{enumerate}

    Let $\theta_1$ be the (unique) injective mapping, such that $\ins_1 = \orig(\ins_1)\theta_1$.

    To prove \ref{lem:equi-label-chains:1}, we want to construct an instance $\ins_2$, such that $c_2\ins_2$ is a chain, which is the case if and only if $\rho_{c_2} \iposr \ins_2$.
    Let $\ins_2 \coloneqq \ins_1 \theta_n^{-1} \theta_m$ be that instance.
    Since $\theta_m$ is injective, we can apply its inverse on both sides and get $\rho_{c_2} \theta_m^{-1} \iposr \ins_2 \theta_m^{-1}$.
    Unfolding $\ins_2$, the right-hand side becomes  $\ins_2 \theta_m^{-1} = \ins_1 \theta_n^{-1} \theta_m \theta_m^{-1} = \ins_1 \theta_n^{-1}$.
    Unpacking Definition~\ref{def:direct-posr} for $\rho_{c_2} \theta_m^{-1} \iposr \ins_1 \theta_n^{-1}$, we need to show for database
    $\I_2 \coloneqq (\bodyP(\rho_{c_2}) \cup \head(\rho_{c_2}) \cup \bodyP(\ins_2)) \theta_m^{-1}\omega$, that
    \begin{enumerate}[leftmargin=.9cm,label=(2.\alph*)]
        \item\label{it:c2-alpha} the intersection between the atom sets $(\head(\rho_{c_2}) \setminus \bodyP(\rho_{c_2})) \theta_m^{-1}$ and $\bodyP(\ins_1) \theta_n^{-1}$ is non-empty, and
        \item\label{it:c2-beta} the database $\I_2$ does not satisfy $\exists \vec{v} .\, \head(\ins_1)\theta_n^{-1}\omega_\forall$, and
        \item\label{it:c2-gamma} none of the negative body atoms of $\ins_1\theta_n^{-1}$ are in $\I_2$.
    \end{enumerate}

    First, we derive \ref{it:c2-alpha}:
    Recall that $\labA(c_1) = \labA(c_2)$ means by Equation~\eqref{eq:alpha-label}, that
    the pieces of $\orig(\ins_n)$, which are not contained in $\bodyP(\rho_{c_1})$ under $\theta_n$,
    are exactly the pieces of $\orig(\ins_m)$, which are not contained in $\bodyP(\rho_{c_2})$ under $\theta_m$.
    Also recall that $\head(\rho_{c_1}) = \head(\orig(\ins_n))\theta_n$ and $\head(\rho_{c_2}) = \head(\orig(\ins_m))\theta_m$.
    Therefore, $(\head(\rho_{c_1}) \setminus \bodyP(\rho_{c_1})) \theta_n^{-1} = (\head(\rho_{c_2}) \setminus \bodyP(\rho_{c_2})) \theta_m^{-1}$.
    With this equality, \ref{it:c2-alpha} directly follows from \ref{it:c1-alpha}.

    Next, we derive \ref{it:c2-beta}:
    Expanding Equation~\eqref{eq:beta-label} at $\labB(c_1) = \labB(c_2)$, we get that
    $(\bodyP(\rho_{c_1}) \cup \head(\rho_{c_1})) \omega \models \exists \vec{v} .\, \tilde{\psi} (\theta_n \omega)_\forall$
    if and only if
    $(\bodyP(\rho_{c_2}) \cup \head(\rho_{c_2})) \omega \models \exists \vec{v} .\, \tilde{\psi} (\theta_m \omega)_\forall$
    for all $\exists \vec{v} .\, \tilde{\psi} \in \mathcal{F} \subseteq \subformulae(\pieces(\variants(\R)))$.
    As $\theta_n$ and $\theta_m$ are both injective, we can apply their inverses to learn that
    $\I_1 \models \exists \vec{v} .\, \tilde{\psi} \omega_\forall$
    if and only if
    $\I_2 \models \exists \vec{v} .\, \tilde{\psi} \omega_\forall$.
    Using the forward direction of this equivalence, \ref{it:c2-beta} follows from \ref{it:c1-beta}.
    The defining condition for $\equiv_\mathcal{F}$ is exactly this equivalence, so $\I_1 \equiv_\mathcal{F} \I_2$, which will later be needed to argue for \ref{lem:equi-label-chains:2}.

    Finally, we derive \ref{it:c2-gamma}:
    By Equation~\eqref{eq:gamma-label}, $\labC(c_1) = \labC(c_2)$ says that
    a variant $\check{\rho}$ has a negative body atom falling into $\I_1$ under $\theta_n$ if and only if it has a negative body atom falling into $\I_2$ under $\theta_m$.
    By \ref{it:c1-gamma}, the variant $\orig(\ins_1) = \ins_1 \theta_n^{-1}$ cannot be one of them, so \ref{it:c2-gamma} follows.

    Having \ref{it:c2-alpha}--\ref{it:c2-gamma}, we conclude that $c_2\ins_2$ is a chain.
    Note that $c_2\ins_2$ might not be decoupled, if the variables from $\var(\ins_1) \setminus \var(\head(\ins_n))$, which necessarily are fresh w.r.t. $c_1$, are not also fresh w.r.t. $c_2$.
    If so, just take any injective mapping $\theta' : \stale(c_2\ins_2) \injTo \V \setminus \bigcup_{i\in\set[m]} \var(\ins_i)$, and $c_2 \ins_2\theta'$ is clearly decoupled.
    As $\theta'$ acts only on the variables in $\ins_2$ not shared with the head of $c_2$'s final instance,
    $\rho_{c_2} \iposr \ins_2\theta'$ holds iff $\rho_{c_2} \iposr \ins_2$ (by an argument similar to the proof of Lemma~\ref{lem:trail-gen-dec-chain}),
    i.e. $c_2 \ins_2\theta'$ is a (decoupled) chain.

    To prove \ref{lem:equi-label-chains:2}, we first establish:

    \begin{claim}[$\ddagger$]
        Let $\mathcal{F}$ be a formula set, such that $\subformulae(\variants(\mathcal{F})) = \mathcal{F}$.
        Let $\I_1$, $\I_2$, and $X$ be databases, such that
        $(\symbols(\I_1 \cup \I_2)) \cap \symbols(X) \subseteq \symbols(\mathcal{F})$.
        Then $\I_1 \equiv_\mathcal{F} \I_2$ implies that $\I_1 \cup X \equiv_\mathcal{F} \I_2 \cup X$.
    \end{claim}

    \begin{proof}
        Let databases $\I_1 \equiv_\mathcal{F} \I_2$, and $X$ as well as formula set $\mathcal{F}$ be as in the claim.
        This means that $\I_1$ and $\I_2$ agree on satisfaction of formulae from $\mathcal{F}$ under $\omega_\forall$
        and that $X$ uses only fresh symbols, except for $\symbols(F)$.
        Let $\exists \vec{v}.\tilde{\psi}$ be any formula from $\mathcal{F}$.

        To prove the claim, we need to show that
        either both $\I_1 \cup X \models \exists \vec{v}.\tilde{\psi}\omega_\forall$ and $\I_2 \cup X \models \exists \vec{v}.\tilde{\psi}\omega_\forall$
        or neither $\I_1 \cup X \not\models \exists \vec{v}.\tilde{\psi}\omega_\forall$ nor $\I_2 \cup X \not\models \exists \vec{v}.\tilde{\psi}\omega_\forall$.

        In case that $\I_1 \models \exists \vec{v}.\tilde{\psi}\omega_\forall$,
        then $\I_2$ must also satisfy it because of $\I_1 \equiv_\mathcal{F} \I_2$.
        As satisfaction for pieces is monotone,
        we directly get that $\I_1\cup X \models \exists \vec{v}.\tilde{\psi}\omega_\forall$ and $\I_2\cup X \models \exists \vec{v}.\tilde{\psi}\omega_\forall$.

        Consider the case that $\I_1\not\models \exists \vec{v}.\tilde{\psi}\omega_\forall$.
        By $\I_1 \equiv_\mathcal{F} \I_2$, we must also have $\I_2 \not\models \exists \vec{v}.\tilde{\psi}\omega_\forall$.
        Suppose that $\I_2\cup X \models \exists \vec{v}.\tilde{\psi}\omega_\forall$.
        So there exists a mapping $\mu_\exists : \vec{v} \to \C\cup\N$ such that $\tilde{\psi} \omega_\forall \mu_\exists \subseteq \I_2\cup X$.
        It remains to show that $\I_1\cup X \models \exists \vec{v}.\tilde{\psi}\omega_\forall$.

        The image of $\mu_\exists$ can be separated into three disjoint sets
        $A = \symbols(\mathcal{F})$ and $B = \symbols(X) \setminus A$ and finally $C = \symbols(\I_2) \setminus A$.
        Accordingly, we can split $\vec{v}$ into three disjoint sets $\vec{v}_A \dot{\cup} \vec{v}_B \dot{\cup} \vec{v}_C$
        where $\vec{v}_S = \{ v\in\vec{v} \mid \mu_\exists(v) \in S \}$ for a set $S$.
        And finally, we can split $\mu_\exists$ into three separate mappings
        $\mu_\exists^A : \vec{v}_A \to A$, $\mu_\exists^B : \vec{v}_B \to B$, and $\mu_\exists^C: \vec{v}_C \to C$,
        such that $\mu_\exists = \mu_\exists^A \circ \mu_\exists^B \circ \mu_\exists^C$.\footnote{
            For a function $f: A \to B$, we call $f^{\id}: A \cup B \to A \cup B$ with $x \mapsto f(x)$ for $x\in A$ and $x \mapsto x$ for $x\in B \setminus A$ the identity extension of $f$.
            Given two functions $f: A_1 \to B$ and $g: A_2 \to B$ where $A_1 \cap A_2 = \varnothing$,
            we write $f \circ g : A_1 \dot{\cup} A_2 \to B$ for the composition $f^{\id} \circ g^{\id}$ of their identify extensions.
        }

        As $\tilde{\psi} \omega_\forall \mu_\exists \subseteq \I_2\cup X$,
        we can write $\tilde{\psi}$ as the disjoint union of two sets $\tilde{\psi}_1$ and $\tilde{\psi}_2$,
        such that $\tilde{\psi}_1 \omega_\forall \mu_\exists \subseteq \I_2$ and $\tilde{\psi}_2 \omega_\forall \mu_\exists \subseteq X$.
        Clearly, the image of existentials from the first set is $\mu_\exists[\varE(\tilde{\psi}_1)] \subseteq \symbols(\I_2) = A \cup C$
        and the image of existentials from the second set is $\mu_\exists[\varE(\tilde{\psi}_2)] \subseteq \symbols(X) = A \cup B$.
        So we have $\tilde{\psi}_1 \omega_\forall \mu_\exists^A \mu_\exists^C \subseteq \I_2$ and $\tilde{\psi}_2 \omega_\forall \mu_\exists^A \mu_\exists^B \subseteq X$.

        Let $\check{\rho}$ be a variant with $\exists \vec{v} .\, \tilde{\psi} \in \subformulae(\check{\rho})$.
        Applying the variable mapping $(\mu_\exists^A\omega^{-1}) : \var(\R) \to \var(\R)$ yields the variant $\check{\rho} (\mu_\exists^A \omega^{-1})$.
        The pieces of the formula $\exists \vec{v}_C .\, \tilde{\psi}_1 (\mu_\exists^A\omega^{-1})$ are subformulae of $\check{\rho}(\mu_\exists^A \omega^{-1}) \in \variants(\mathcal{F})$.
        Since each of them is satisfied by $\I_2$ (with $\mu_\exists^C$) they must also be satisfied by $\I_1$,
        i.e. we must be able to find a $\mu_\exists^{C'}: \vec{v}_C \to \symbols(\I_1)$,
        such that $\tilde{\psi}_1 (\mu_\exists^A \omega^{-1}) \omega_\forall \mu_\exists^{C'} = \tilde{\psi}_1 \omega_\forall \mu_\exists^A \mu_\exists^{C'} \subseteq \I_1$.
        Thus, $\I_1 \cup X \models \exists \vec{v} .\, \tilde{\psi}\omega_\forall$ with $\mu_\exists^A \mu_\exists^B \mu_\exists^{C'}$.
\qed
    \end{proof}

    For \ref{lem:equi-label-chains:2}, we need to show that $\ell(c_1\ins_1) = \ell(c_2\ins_2)$.
    We will show this for the three components of the label separately.

    Clearly, we have $\labA(c_1\ins_1) = \labA(c_2\ins_2)$,
    because the final instances $\ins_1$ and $\ins_2$ both correspond to the same variant $\orig(\ins_1)$.
    In \ref{it:c1-alpha} and \ref{it:c2-alpha} we thus identified the same non-empty intersection between the body of this variant and the heads of the chain rules minus their bodies.

    As noted, $\labB(c_1) = \labB(c_2)$ gives $\I_1 \equiv_\mathcal{F} \I_2$.
    Since $\subformulae(\variants(\cdot))$ is idempotent
    and the set $\mathcal{F}$ used for computing $\labB$ is
    $\subformulae(\mathcal{H})$ (or with $\RpD$-closure $\subformulae(\mathcal{H}\cup\mathcal{B})$),
    and $\mathcal{H}$ and $\mathcal{B}$ are closed under $\variants(\cdot)$,
    it fulfils the requirement of the claim $\ddagger$.
    The variables used universally in $\mathcal{F}$ are from $\var(\R)$, i.e. $X$ may share only symbols $\var(\R)\omega$ with the union of $\I_1$ and $\I_2$.
    Note that $\symbols(\I_1\cup\I_2) = (\symbols(\I_1) \cup \symbols(\I_2))$.
    The set of facts $X = (\bodyP(\orig(\ins_1)) \cup \head(\orig(\ins_1))) \omega$ meets the conditions for claim $\ddagger$,
    since the instance is decoupled w.r.t. both chains $c_1$ and $c_2$.
    So we find that $\I_1 \cup X \equiv_\mathcal{F} \I_2 \cup X$ by ($\ddagger$).
    This set $X$ is exactly the set of facts that is added to the respective databases for computing $\labB(c_1)$ and $\labB(c_2)$.

    Lastly, we clearly have $\labC(c_1\ins_1) = \labC(c_2\ins_2)$, as the fact set $X$ which adds to the databases are again identical.
\qed
\end{proof}

\thmLrRegular*

\begin{claim}[$\dagger$]
    Let $w_1,w_2\in \L(\R)$ be two words with the same final letter.
    For any $v\in\Alphabet^*$, $w_1v\in \L(\R)$ iff $w_2v\in \L(\R)$.
\end{claim}

\begin{proof}[by induction on $|v|$]
    For the induction base, the length of $v$ is zero, i.e. $v = \epsilon$ (the empty word) and $w_2v = w_2$. By assumption $w_2 \in \L(\R)$ and so directly $w_2v \in \L(\R)$.

    The induction hypothesis (IH) is that for any word $v'\in\Alphabet^*$ and symbol $c \in \Alphabet$,
    $w_1v'\in \L(\R) \iff w_2v'\in \L(\R)$ implies $w_1v'c\in \L(\R) \iff w_2v'c\in \L(\R)$.

    For the induction step, let $v = v'c$ for $v'\in\Alphabet^*$ and $c\in\Alphabet$.
    If $v' = \epsilon$, the final letter of $w_1v'$ is the final letter of $w_1$ and the final letter of $w_2v'$ is the final letter of $w_2$, which are the same by assumption.
    Otherwise, the final letter of both $w_1v'$ and $w_2v'$ is the final letter of $v'$.
    Therefore, $w_1v'$ and $w_2v'$ have identical final letters in both cases and hence correspond to equi-labelled chains.
    Applying Lemma~\ref{lem:equi-label-chains} finishes the induction.
\qed
\end{proof}

\begin{itemize}
    \item
    For sets $S$ and $T$, an equivalence relation $\sim$ over $S$ distinguishes elements based on a predicate $p: S \to T$, if $x \sim y$ iff $p(x) = p(y)$ for all $x,y\in S$.
    \item
    For a language $\L() \subseteq \Alphabet^*$, the extensions of a word $w \in \L()$ is the set of all $v\in\Alphabet^*$ for which $wv \in \L()$.
    \item
    By Myhill-Nerode, a language is regular, if it is possible to give an equivalence relation over its words
    that has finitely many equivalence classes and distinguishes
    words based on their extensions.
\end{itemize}

\begin{proof}[of Theorem~\ref{thm:LR-regular}]
    We define an equivalence relation over words in the language $\L(\R)$ based on their final letters
    \begin{equation}
        \sim_{\L(\R)} \coloneqq \{ \langle w(c_1),w(c_2) \rangle \mid c_1,c_2\in\chains(\R), \ell(c_1)=\ell(c_2) \}\text{.}
    \end{equation}
    Clearly, $\sim_{\L(\R)}$ has at most $|\Alphabet|$-many equivalence classes, which is finitely many as the alphabet is finite.
    By ($\dagger$), two words from $\L(\R)$ that share the same label ($\hat{=}$ final letter), i.e. are $\sim_{\L(\R)}$-equivalent, have the same extensions.

    Hence, ${\sim_{\L(\R)}} \cup (\Alphabet^*\setminus\L(\R))^2$ meets the conditions of Myhill-Nerode, which makes $\L(\R)$ a regular language.
\qed
\end{proof}

\begin{claim}\label{lem:equi-label-chains-restr-negr}
    For $c_1,c_2 \in \chains(\R)$ with $\ell(c_1)=\ell(c_2)$,
    we have $\rho_{c_1} \restr \rho$ iff $\rho_{c_2} \restr \rho$ and $\rho_{c_1} \negr \rho$ iff $\rho_{c_2} \negr \rho$.
\end{claim}

\begin{proof}
    To compute $\rho_{c_i}\restr\rho$ (see Definition~\ref{def:restr}) and $\rho_{c_i}\negr\rho$ (see Definition~\ref{def:negr}) for $i\in\{1,2\}$,
    it suffices to examine unifiable subsets of $\head(\rho)$ or $\bodyN(\rho)$ with $\head(\rho_{c_i})$
    (c.f. efficient algorithms of \cite{GIKM2022}).
    Remaining checks concern the satisfaction of the heads of the involved rules on representative databases.
    Let $\ins_{n_i}$ be the final instances in $c_i$, with unique $\theta_{n_i}$, such that $\ins_{n_i} = \orig(\ins_{n_i})\theta_{n_i}$.
    Again, we can equivalently test for $\rho_{c_i}\theta_{n_i}^{-1}\restr\rho\theta_{n_i}^{-1}$ and
    $\rho_{c_i}\theta_{n_i}^{-1}\negr\rho\theta_{n_i}^{-1}$ (as $\theta_{n_i}$ is revertible and $\theta_{n_i}^{-1}$ is injective).
    These will involve $(\bodyP(\rho_{c_i}) \cup \head(\rho_{c_i}))\theta_{n_i}^{-1}\omega$,
    and we stated before that these are $\equiv_\mathcal{F}$-equivalent due to $\labB(c_1)=\labB(c_2)$
    (see proof of Lemma~\ref{lem:equi-label-chains}~\ref{lem:equi-label-chains:1}) and satisfy the same rule heads.
\qed
\end{proof}
 
\clearpage
\section{Details on Introductory Example}
\def\t(#1){\allowbreak t(#1)}

\vspace{1em}
\noindent
We will now revisit the introductory examples (cf. Section~\ref{sec:intro}, Problems~1 and 2).
Here, we refer to the N3 rule in line $i$ by $\rho_i$.
Figure~\ref{fig:intro_rel-graphs} depicts a graphical representation of the reliance relations between these rules.
Evidently, this yields an acyclic graph for \alglineref{rule_father}--\alglineref{rule_different_fathers} on the right,
which means that $\R_{\text{P2}} = \{ \rho_8,\rho_9,\rho_{10},\rho_{11} \}$ is \emph{fully stratified} (and thus also \emph{chain-stratified}).
The graph on the left, however, has cycles through negative reliances $\rho_1\negr\rho_2$ and $\rho_4\negr\rho_2$,
i.e. $\R_{\text{P1}} = \{ \rho_1,\rho_2,\rho_3,\rho_4 \}$ and $\R'_{\text{P1}} = \R_{\text{P1}} \cup \{ \rho_5,\rho_6,\rho_7 \}$ are not \emph{fully stratified}.
Below, we will examine $\R_{\text{P1}}$ and $\R'_{\text{P1}}$ in detail, to see that
--- while neither is \emph{chain-stratified} either (see Example~\ref{ex_problem1_solved}) ---
$\R'_{\text{P1}}$ is \emph{chain-stratified under constraints} (as defined in Section~\ref{sec_refinements}).

\begin{figure}
    \centering
    \begin{tikzpicture}[>={Latex[length=2mm,width=1.5mm]}, baseline=(current bounding box.north)]
        \node[draw,circle] (1) at (0,0) {$\rho_1$};
        \node[draw,circle] (2) at (0,-2) {$\rho_2$};
        \node[draw,circle] (3) at (2,0) {$\rho_3$};
        \node[draw,circle] (4) at (2,-2) {$\rho_4$};
        \node[draw,circle] (5) at (4,0) {$\rho_5$};
        \node[draw,circle] (6) at (4,-1) {$\rho_6$};
        \node[draw,circle] (7) at (4,-2) {$\rho_7$};

        \draw[->] (1)--node[left]{$\negr$}(2);
        \draw[->] (4)edge[bend left]node[below]{$\negr$}(2);

        \draw[dashed,->] (3)edge[loop left](3);
        \draw[dashed,<->] (3)--(4);
        \draw[dashed,->] (3)--(5);
        \draw[dashed,->] (3)--(6);
        \draw[dashed,->] (3)--(7);
        \draw[dashed,<->] (4)--(1);
        \draw[dashed,<->] (4)--(2);
        \draw[dashed,<->] (4)edge[loop below](4);
        \draw[dashed,->] (4)--(5);
        \draw[dashed,->] (4)--(6);
        \draw[dashed,->] (4)--(7);
    \end{tikzpicture}\hspace{5em}
    \begin{tikzpicture}[>={Latex[length=2mm,width=1.5mm]}, baseline=(current bounding box.north)]
        \node[draw,circle] (8) at (0,0) {$\rho_8$};
        \node[draw,circle] (9) at (0,-2) {$\rho_9$};
        \node[draw,circle] (10) at (2,0) {$\rho_{10}$};
        \node[draw,circle] (11) at (2,-2) {$\rho_{11}$};

        \draw[->] (10)--node[right]{$\negr$}(11);
        \draw[->] (9)--node[left]{$\restr$}(8);

        \draw[dashed,->] (8)--(10);
        \draw[dashed,->] (8)--(11);
    \end{tikzpicture}
    \caption{Reliance graph of \alglineref{rule_ex_intro_student}--\alglineref{rule_constraint_teacher_type} (left) and \alglineref{rule_father}--\alglineref{rule_different_fathers} (right); dashed lines are $\posr$}\label{fig:intro_rel-graphs}
\end{figure}

\noindent
Let us first rewrite L1--L7 in more compact syntax and using fewer variables.
\begin{equation*}
    \begin{array}{cc}
        \begin{aligned}
            \rho_1&:\quad \t(?x,pa,?y) \ruleTo \t(?y,ty,St) \\
            \rho_2&:\quad \t(?x,te,?y), \neg \t(?y,ty,St) \ruleTo \t(?x,ex,?y) \\
            \rho_3&:\quad \t(?p,spo,?y), \t(?y,spo,?q) \ruleTo \t(?p,spo,?q) \\
            \rho_4&:\quad \t(?p,spo,?q), \t(?x,?p,?y) \ruleTo \t(?x,?q,?y)
        \end{aligned}
        &
        \begin{aligned}
            \rho_5&:\quad \t(ex,spo,pa) \ruleTo \bot \\
            \rho_6&:\quad \t(ex,spo,spo) \ruleTo \bot \\
            \rho_7&:\quad \t(ex,spo,ty) \ruleTo \bot
        \end{aligned}
    \end{array}
\end{equation*}
\noindent
There are cycles through negative reliances, e.g. $\rho_2 \posr \rho_4 \posr \rho_1 \negr \rho_2$ or just $\rho_2 \posr \rho_4 \negr \rho_2$.
All $\negr$-edges end in $\rho_2$, so we only need to examine chains starting with instances of $\rho_2$.
As explained earlier, the ruleset is indeed not \emph{chain stratified},
but it is \emph{chain stratified under constraints} due to the constraints $\rho_5$, $\rho_6$ and $\rho_7$.
Algorithm~\ref{alg:chains} will consider the following chains in the given order:

\vspace{1em}
{\raggedright \fontsize{8pt}{10pt}\selectfont \begin{description}
    \item[$\rho_2,\rho_4\>:\>$]
    $\t(?x,te,?y), \t(?x,ex,?y), \t(ex,spo,?q_1) \ruleTo \t(?x,?q_1,?y)$
    \begin{itemize}
        \item[$\scriptstyle\bullet$] chain accepted
        \item[$\scriptstyle\bullet$] $\negr\rho_2$ rejected because it would require $?q_1 \mapsto ty$, $?y\mapsto St$, violating constraint $\rho_7$
    \end{itemize}
    \item[$\rho_2,\rho_4,\rho_1\>:\>$]
    $\t(?x,te,?y), \t(?x,ex,?y), \t(ex,spo,pa), \t(?x,pa,?y) \ruleTo \t(?y,ty,St)$
    \begin{itemize}
        \item[$\scriptstyle\bullet$] chain rejected because of constraint $\rho_5$
    \end{itemize}
    \item[$\rho_2,\rho_4,\rho_2\>:\>$]
    $\t(?x,te,?y), \t(?x,ex,?y), \t(ex,spo,te), \t(?x,te,?y) \ruleTo \t(?x,ex,?y)$
    \begin{itemize}
        \item[$\scriptstyle\bullet$] chain rejected because the head is already entailed
    \end{itemize}
    \item[$\rho_2,\rho_4,\rho_3\>:\>$]
    $\t(?x,te,?y), \t(?x,ex,?y), \t(ex,spo,spo), \t(?x,spo,?y), \t(?y,spo,?z) \ruleTo \t(?x,spo,?z)$
    \begin{itemize}
        \item[$\scriptstyle\bullet$] chain rejected because of constraint $\rho_6$
    \end{itemize}
    \item[$\rho_2,\rho_4,\rho_4\>:\>$] \phantom{.}
    \begin{description}
        \item[{$\left[?p/?q_1,?q/?q_2\right]\>:\>$}] $\t(?x,te,?y), \t(?x,ex,?y), \t(ex,spo,?q_1), \t(?x,?q_1,?y), \t(?q_1,spo,?q_2) \ruleTo \t(?x,?q_2,?y)$
        \begin{itemize}
            \item[$\scriptstyle\bullet$] chain accepted
            \item[$\scriptstyle\bullet$] $\negr\rho_2$ rejected because it would require $?q_2 \mapsto ty$, $?y\mapsto St$, such that $\rho_3$ derives $\t(ex,spo,ty)$, violating constraint $\rho_7$
        \end{itemize}
        \item[{$\left[?q_1/spo,?p/?x,?q/?y\right]\>:\>$}] $\t(?x,te,?y), \t(?x,ex,?y), \t(ex,spo,spo), \t(?x,spo,?y), \t(?x_1,?x,?y_1) \ruleTo \t(?x_1,?y,?y_1)$
        \begin{itemize}
            \item[$\scriptstyle\bullet$] chain rejected because of constraint $\rho_6$
        \end{itemize}
    \end{description}
    \item[$\rho_2,\rho_4,\rho_4,\rho_1\>:\>$]
    $\t(?x,te,?y), \t(?x,ex,?y), \t(ex,spo,?q_1), \t(?x,?q_1,?y), \t(?q_1,spo,pa), \t(?x,pa,?y) \ruleTo \t(?y,ty,St)$
    \begin{itemize}
        \item[$\scriptstyle\bullet$] chain rejected because $\rho_3$ derives $\t(ex,spo,pa)$, violating constraint $\rho_5$
    \end{itemize}
    \item[$\rho_2,\rho_4,\rho_4,\rho_2\>:\>$]
    $\t(?x,te,?y), \t(?x,ex,?y), \t(ex,spo,?q_1), \t(?x,?q_1,?y), \t(?q_1,spo,te), \t(?x,te,?y) \ruleTo \t(?x,ex,?y)$
    \begin{itemize}
        \item[$\scriptstyle\bullet$] chain rejected because the head is already entailed
    \end{itemize}
    \item[$\rho_2,\rho_4,\rho_4,\rho_3\>:\>$]
    $\t(?x,te,?y), \t(?x,ex,?y), \t(ex,spo,?q_1), \t(?x,?q_1,?y), \t(?q_1,spo,spo), \t(?x,spo,?y), \t(?y,spo,?z) \ruleTo \t(?x,spo,?z)$
    \begin{itemize}
        \item[$\scriptstyle\bullet$] chain rejected because $\rho_3$ derives $\t(ex,spo,spo)$, violating constraint $\rho_6$
    \end{itemize}
    \item[$\rho_2,\rho_4,\rho_4,\rho_4\>:\>$] \phantom{.}
    \begin{description}
        \item[{$\left[?p/?q_2,?q/?q_3\right]\>:\>$}] $\t(?x,te,?y), \t(?x,ex,?y), \t(ex,spo,?q_1), \t(?x,?q_1,?y), \t(?q_1,spo,?q_2), \t(?x,?q_2,?y), \t(?q_2,spo,?q_3) \ruleTo \t(?x,?q_3,?y)$
        \begin{itemize}
            \item[$\scriptstyle\bullet$] chain accepted (...but continuing like this will repeat label)
            \item[$\scriptstyle\bullet$] $\negr\rho_2$ rejected because it would require $?q_3 \mapsto ty$, $?y\mapsto St$, such that two $\rho_3$ applications derive $\t(ex,spo,ty)$, violating constraint $\rho_7$
        \end{itemize}
        \item[{$\left[?q_2/spo,?p/?x,?q/?y\right]\>:\>$}] $\t(?x,te,?y), \t(?x,ex,?y), \t(ex,spo,?q_1), \t(?x,?q_1,?y), \t(?q_1,spo,spo), \t(?x,spo,?y), \t(?x_1,?x,?y_1) \ruleTo \t(?x_1,?y,?y_1)$
        \begin{itemize}
            \item[$\scriptstyle\bullet$] chain rejected because $\rho_3$ derives $\t(ex,spo,spo)$, violating constraint $\rho_6$
        \end{itemize}
    \end{description}
\end{description}
}

None of these passed the conditions in A\ref{line_new_restr} or A\ref{line_new_negr} of the algorithm, so no new edges were added to $\bm{\prec}$.
Therefore, A\ref{line_return_false} will not find cycles and A\ref{line_return_true} is reached, concluding that $\R'_{\text{P1}}$ indeed is \emph{chain-stratified under constraints}.

\fi

\end{document}